\documentclass[a4paper,UKenglish]{lipics-v2018}

\nolinenumbers

\usepackage{microtype}

\usepackage{url}
 
\usepackage{xcolor}

\newcommand{\commenttext}[1]{ \begin{center} {\fbox{\begin{minipage}[h]{0.9
						\linewidth}   {\sf #1} \end{minipage} }} \end{center}}

\newcommand\VRAIMENTFINAL[1]{
}
\newcommand\FINAL[1]{}

\newcommand{\olivier}[1]{\VRAIMENTFINAL{{\color{red}\commenttext{Olivier: #1}}}}
\newcommand{\sabrina}[1]{\VRAIMENTFINAL{{\color{blue}\commenttext{Sabrina: #1}}}}

\newcommand\olivierOK[1]{\FINAL{\olivier{#1}}}
\newcommand\sabrinaOK[1]{\FINAL{\sabrina{#1}}}

\newcommand\mytitre{Cheap Non-standard Analysis and Computability }

\title{\mytitre}

\author{Olivier Bournez}{Ecole Polytechnique, LIX, 91128 Palaiseau Cedex, France}{bournez@lix.polytechnique.fr}{}{ANR PROJECT RACAF}

\author{Sabrina Ouazzani}{LACL, Universit\'e Paris-Est Cr\'eteil, 61 avenue du
  G\'en\'eral de Gaulle, 94010 Cr\'eteil,
  France}{sabrina.ouazzani@lacl.fr}{}{}

\authorrunning{O. Bournez and S. Ouazzani} 

\subjclass{ Theory of computation/Computability, Mathematics of
  computing/Numerical analysis, Models of
  Computation/Continuous functions. 
}

\usepackage[utf8]{inputenc}
\usepackage[T1]{fontenc}

\usepackage[cyr]{aeguill}
\usepackage{url}

\newcommand\vn{\vectorl{n}}

\newcommand\R{\mathbb{R}}
\newcommand\N{\mathbb{N}}
\newcommand\Q{\mathbb{Q}}
\newcommand\QP{\mathbb{Q}^{>0}}
\newcommand\D{\mathbb{D}}

\newtheorem{acks}{ACKS}

\usepackage{color}

\newcommand\partientieresup[1]{\lceil #1 \rceil}

\begin{document}

\title{Cheap Non-standard Analysis and Computability}         
\keywords{Computability Theory, Computable Analysis, Non-Standard Analysis}  

\maketitle

\olivierOK{ To Sabrina:il faudrait compléter/rendre pertinant dans la macro subjclass ce qui
  est approprié. Je cite `
Il ya d ans le répertoire lipics-v2018* le style lipics et surement
quelquepart des explications de ce qui est attendu.
}

\begin{abstract}
  Non standard analysis is an area of Mathematics dealing with notions
  of infinitesimal and infinitely large numbers, in which many
  statements from classical analysis can be expressed very naturally.
  Cheap non-standard analysis introduced by Terence Tao in 2012 is based on
  the idea that considering that a property holds eventually is
  sufficient to give the essence of many of its statements.  This
  provides constructivity but at some (acceptable) price.

  We consider computability in cheap non-standard analysis. We prove
  that many concepts from computable analysis as well as several
  concepts from computability can be very elegantly and alternatively
  presented in this framework. Our statements provide the bases for  dual views and dual proofs
  to several
  statements already known in these fields.

\olivierOK{MFCS: enlevé. The model is demonstrated to also cover
  ordinal computability models.}
\end{abstract}

\olivierOK{To Sabrina: Il faut intégrer les commentaires (au minimum toutes les
  typos des referees de LICS. }

\section{Introduction}
\label{sec:introduction}


\newcommand\std[1]{ \textcolor{cyan}
  {{{#1}}}}
\newcommand\nstd[1]{ \textcolor{red}
  {{\mathbf{#1}}}}
\newcommand\idx[1]{ \textcolor{blue}
  {{#1}}}
\newcommand\val[1]{ \textcolor{blue}
  {{#1}}}

\newcommand\SFT[1]{\overline{#1}}
\newcommand\Snepsilon{\SFT{\nepsilon}}
\newcommand\Sndelta{\SFT{\ndelta}}
\newcommand\Snp{\SFT{\np}}
\newcommand\Snq{\SFT{\nq}}

\newcommand\myomega{\underline{\omega}}
\newcommand\mylambda{\underline{\lambda}}

\newcommand\sN{\std{N}}
\newcommand\sK{\std{K}}
\newcommand\sx{\std{x}}
\newcommand\sy{\std{y}}
\newcommand\sz{\std{z}}
\newcommand\sn{\std{n}}
\newcommand\sk{\std{k}}
\newcommand\sm{\std{m}}
\newcommand\sr{\std{r}}
\newcommand\sa{\std{a}}
\newcommand\ssb{\std{b}}
\newcommand\ssc{\std{c}}
\newcommand\sd{\std{d}}
\newcommand\slambda{\std{\lambda}}
\newcommand\sss{\std{s}}
\newcommand\ssf{\std{f}}
\newcommand\sF{\std{F}}
\newcommand\sG{\std{G}}
\newcommand\sg{\std{g}}
\newcommand\sh{\std{h}}
\newcommand\ssp{\std{p}}
\newcommand\sq{\std{q}}
\newcommand\smu{\std{\mu}}
\newcommand\sepsilon{\std{\epsilon}}
\newcommand\sdelta{\std{\delta}}
\newcommand\spsi{\std{\psi}}
\newcommand\si{\std{i}}
\newcommand\sX{\std{X}}
\newcommand\sY{\std{Y}}
\newcommand\sD{\std{D}}
\newcommand\sP{\std{P}}

\newcommand\nN{\nstd{N}}
\newcommand\npsi{\nstd{\psi}}
\newcommand\na{\nstd{a}}
\newcommand\nb{\nstd{b}}
\newcommand\nx{\nstd{x}}
\newcommand\ny{\nstd{y}}
\newcommand\nd{\nstd{d}}
\newcommand\nz{\nstd{z}}
\newcommand\nf{\nstd{f}}
\newcommand\nk{\nstd{k}}
\newcommand\nn{\nstd{n}}
\newcommand\np{\nstd{p}}
\newcommand\nq{\nstd{q}}
\newcommand\nm{\nstd{m}}
\newcommand\nnu{\nstd{u}}
\newcommand\nX{\nstd{{}^*X}}
\newcommand\nY{\nstd{{}^*Y}}
\newcommand\nD{\nstd{{}^*D}}
\newcommand\nS{\nstd{{}^*S}}
\newcommand\csomega{\std{\myomega}}
\newcommand\comega{\nstd{\myomega}}
\newcommand\clambda{\nstd{\mylambda}}
\newcommand\nepsilon{\nstd{\epsilon}}
\newcommand\ndelta{\nstd{\delta}}


\newcommand\nSn{\nstd{{}^*S}_{\iin}}
\newcommand\nNn{\nstd{N}_{\iin}}
\newcommand\nxi{\std{x}_{\ii}}
\newcommand\nxn{\std{x}_{\iin}}
\newcommand\ndn{\std{d}_{\iin}}
\newcommand\nyi{\std{y}_{\ii}}
\newcommand\nyn{\std{y}_{\iin}}
\newcommand\nzn{\std{z}_{\iin}}
\newcommand\nfn{\std{f}_{\iin}}
\newcommand\nXn{\std{X}_{\iin}}
\newcommand\nYn{\std{Y}_{\iin}}
\newcommand\nmn{\std{m}_{\iin}}
\newcommand\nnni{\std{n}_{\ii}}
\newcommand\npn{\std{p}_{\iin}}
\newcommand\nqn{\std{q}_{\iin}}
\newcommand\nepsilonn{\nstd{\epsilon}_{\iin}}
\newcommand\ndeltan{\nstd{\delta}_{\iin}}
\newcommand\comegan{\nstd{\myomega}_{\iin}}

\newcommand\im{\idx{m}}
\newcommand\iin{\idx{n}}
\newcommand\ii{\idx{i}}
\newcommand\ik{\idx{k}}

\renewcommand\vn{\val{n}}
\newcommand\vi{\val{i}}


\newcommand\A{\mathcal{A}}

\newcommand\cheap{cheap non-standard }
\newcommand\cheapP{cheap non-standard}
\newcommand\Cheap{Cheap non-standard }
\newcommand\CheapM{Cheap Non-Standard }

\newcommand\gammaun{\gamma}


\newcommand\shift {shift }
\newcommand\Shift {Shift }
\newcommand\shiftu[1]{#1^{+}}
\newcommand\shiftd[2]{#1^{+#2}}

\newcommand\machintrucby {effective with respect to }
\newcommand\machinequivalent{computably equivalent }
\newcommand\machinequivalentn{computably equivalent}
\newcommand\ieffective{effective } 
\newcommand\ieffectiven{effective} 
\newcommand\effectivelyuniformlycontinous{has an effective modulus of continuity
}


\newcommand\finite{{finite} }
\newcommand\finitely{\emph{finitely} }
\newcommand\FrechetFilter{Fr\'echet filter}
\newcommand\totalrecursive{\emph{total recursive }}
\newcommand\totalnonrecursive{\emph{total non recursive }}
\renewcommand\succ{\sss}



\newcommand\Idx{\mathcal{N}}
\newcommand\metafinite{meta-finite}
\newcommand\Metafinite{Meta-finite}
\newcommand\I{\mathcal{I}}
\newcommand\patchproperty{\emph{patch property}}
\newcommand\NN{{\N}}

\newcommand\Los{\L{}\'os}

\newcommand\integer{{integer}}
\newcommand\integers{{integers}}


\newcommand\recursivelyregular{recursively regular }


\newcommand\ov[1]{\overline{#1}}

\olivier{MFCS: mieux distinguer entier de vecteurs d'entiers}


\olivierOK{
Conventions (qu'il faudrait tenir):
\begin{itemize}
\item on dit ``total recursive'' pour une fonction, et pas
  ``computable''
\item par contre, computable pour les cheap non-standard integer
\item lettre minuscule pour les cheap non-standard integer
\item pour dire que les composantes sont $f(i)$, on dit que $\nx=\nxi=f(\ii)$.
\item Macros:
  \begin{itemize}
  \item \Cheap, and \cheap
  \item $\A$
  \item $\comega$ pour $\comega_i=i$
  \item $\gammaun$
  \end{itemize}

\end{itemize}
}

\olivierOK{Nouveautes de la version 3:
  \begin{itemize}
  \item Des macros pour le vocabulaire, car il est perfectible:
    \begin{itemize}
    \item \shift
    \item \machintrucby, \machinequivalent, \ieffective, \effectivelyuniformlycontinous
    \end{itemize}
  \item la notion de shift est une macro (et changee)
  \end{itemize}

}

\olivierOK{Nouveautes de la version 5:
  \begin{itemize}
  \item La subtilité de ce qui est appelé ``finite''.
  \item Macros \FrechetFilter, \totalrecursive, \succ
  \end{itemize}

}

\olivierOK{Nouveautes de la version 6:
  \begin{itemize}
  \item Un ensemble d'indice $\Idx$. 
   \item Macro \metafinite, \patchproperty
  \item $\NN$
  \end{itemize}
}

\olivier{changé l'ordre du texte.}

While historically reasonings in mathematics were often based on the
used of infinitesimals, in order to avoid paradoxes and debates about
the validity of some of the arguments, this was later abandoned in
favor of epsilon-delta based definitions such as today's classical definition of
continuity for functions over the reals.

Non standard analysis (NSA) originated from the work of Abraham Robinson in
the 1960's who came with a formal construction of non-standard models
of the reals and of the integers \cite{robinson1996non}. Many
statements from classical analysis can be expressed very
elegantly, using concepts such as infinitesimals or infinitely large
numbers in NSA: See
e.g. \cite{robinson1996non,goldblatt2012lectures,keisler1976foundations,keislerbook}.
%
%
%
It not only have interests for understanding
historical arguments and the way we came to some of today's notions,
but also clear interests for pedagogy and providing results that have
not been obtained before in Mathematics. See
e.g. \cite{keisler1976foundations} for nice presentations of
NSA, or \cite{keislerbook} for an undergraduate
level book presenting in a very natural way the whole
mathematical calculus, based on Abraham Robinson's infinitesimals
approach. See \cite{NSAteaching} for recent
instructive pedagogical experiments on its help for teaching mathematical concepts
to students. 

However, the construction and understanding of concepts from
NSA is sometimes hard to grasp. Its models are
built using concepts such as ultrafilter that are obtained using
non-constructive arguments through the existence of a free ultrafilters
whose existence requires to use the axiom of choice.  Moreover, 
the dependance on
the choice of this ultrafilter is sometimes not easy to understand (at
least for non model-theory experts).

Terence Tao came in 2012 in a post in his blog
\cite{CheapNonStandardAnalysis} with a very elegant explanation of the
spirit of many of the statements of non-standard analysis using
only very simple arguments, that he called \cheap analysis in
opposition to classical non-standard analysis. This theory is based on
the idea that the asymptotic limit of a sequence given by its value
after some \finite rank is enough to define non standard
objects
. 
\Cheap analysis provides constructivity but this of course comes with
some price (e.g. a non-total
order on \cheap integers, i.e. some indeterminacy). 


Computability theory, classically dealing with finite or discrete
structures such as a finite alphabet or the integers, has been far
extended in many directions at this date. Various approaches have been
considered for formalizing its issues in analysis, but at this stage
the most standard approaches for dealing with computations over the
reals are from computable analysis \cite{Wei00} and computability for
real functions \cite{Ko91}. For other approaches
for modeling computations over the reals and how they relate,
see e.g. \cite{BCSS} 
or the
appendices of \cite{Wei00}.

\olivierOK{Il faut insérere d'autres refs pertinentes,  dont au moins celle pointée par referees ICALP}

\olivier{Revoir résultats quand aura terminé}

In this paper, we explore how computability mixes with \cheap
 analysis. We prove that many concepts from computable
analysis as well as several concepts from computability can be very
elegantly and alternatively presented in this framework. In
particular, we prove that computable analysis concepts have very
nice and simple formulations in this framework. We also obtain alternative,
equivalent and nice formulations of many of its concepts in this
framework.


Our approach provides an alternative to the usual presentation of
computable analysis. In particular, nowadays, a popular approach to formalize computable analysis is based on
Type-2 computability, i.e. Turing machines working over
representations of objects by infinite sequences of integers: See
\cite{Wei00} for a monograph based on this approach. Other presentations include original
ones from \cite{Tur36short}, \cite{Grz57},
and \cite{Lac55}. 
%
More recently, links have been established between type-2
computability and transfinite computations (see \cite{GNCompAn} for
example) using surreal numbers. 
 NSA has also been used in the context of
various applications like systems modeling: See e.g. 
\cite{proginf} or \cite{benveniste2012non}. 


\olivier{Plans? Revoir plan quand aura terminé}

The paper is organized as follows. In Section \ref{CNSA}, we recall
\cheap analysis. In Section \ref{filter}, we present the very basics
of constructions from NSA, and we state some relations to \cheap analysis. 
In
Section \ref{sec:computabilityintegers}, we start to discuss
computability issues, and we consider computability
of \cheap integers and rational numbers. In Section
\ref{sec:infinitesimals}, we discuss some computability issues related
to infinitesimals and infinitely large numbers. In Section
\ref{sec:reals}, we discuss computability for real numbers.  In
Section \ref{sec:continuity}, we go to computability for functions
over the reals and we discuss continuity and uniform continuity. In
Section \ref{sec:functions}, we discuss computability of functions
over the reals. 
In Section \ref{sec:example}, we discuss some applications illustrating the
interest of using our framework. Finally, in Section \ref{sec:discuss},
we discuss our constructions and we discuss some interesting perspectives. 

\olivier{En fait, $\omega$ c'est pour les indices, et $\N$ pour les
  valeurs, meme si c'est la meme chose.}
\olivierOK{J'ai enlevé l'
  histoire de numbers. Vérifier que pas d'autres references à numbers}

In all what follows, $\N$, $\Q$ and $\R$ are respectively the set of
integers, rational numbers and real numbers.  We will sometimes also
write $\omega$ for a synonym for $\N$: we will use $\omega$ preferably
when talking about indices. 
In what follows, $\A$ is either $\N$ or $\Q$, and we assume $\N \subset \Q$. By
number, we mean either a natural or a rational number.
In the current
version, we use a color
coding to help our reader to visualize the type of each
variable. However, the paper can be read 
without this
color coding.

\olivierOK{Etait: We will consider some index set $\Idx$ that we will
consider up to Section \ref{sec:generalisation} to be $\N$. Similarly,
up to Section \ref{sec:generalisation}, when we say a \finite $i$, we
mean an $i \in \N$ and $\NN$ is a synonym for $\N$. }

\olivierOK{

COMMENTE \% OLD MACROS

\olivierOK{je dis admissible dans la phrase qui suit. admissible?}

\olivierOK{Change V10: admissible -> recursively regular. $\lambda$
  devenu $\kappa$.}
\olivierOK{MFCS: en fait  admissible. Maois bon, a voir si on garde les ordinaux}

Devenu:}

\olivier{MFCS: Commente pour l'instant: A DISCUTER CE QU'ON DIT AU
  FINAL DANS CETTE VERSION MFCS.
Actually, as discussed in Section \ref{sec:generalisation}, $\Idx$ can
actually be some greater ordinal containing $\N$ in purpose to extend
the sequences indexing. In order to distinguish between
$\N$ and $\NN$ we will respectively write natural number and \integer
for the generalised index. See Section \ref{sec:generalisation} for a
full discussion.

The generalised index is motivated by our willingness not only to link
+computability to cheap non-standard analysis (CNSA) but also to
+higher-order computability, in particular to the generalized
+computability defined along well-closed sets such that the admissible
+ordinals (see Sacks' higher recursion theory for example,
\cite{sacks1990higher}). We believe that infinitely large numbers
introduced in CNSA but extended to indices belonging to some
well-chosen sets are of some help to establish interesting parallels
between higher-order recursion models and what is called here
``meta-finite'' recursion. This leads to further perspectives, for
instance a new paradigm equivalent, to some+extent, to certain
existing transfinite time computation models, hence introducing a more
mathematical and thus easier to manipulate model than the current
ones.  Here we allow ourselves some speculation on the possibilities
offered by the proposed framework to further explore computability
theory applied to reals and higher-order computability (and their
relations), this paper provides the basis for these developments.
}

\olivierOK{Commente:
We say that a limit ordinal $\kappa$ is \recursivelyregular if and
only if it is closed under all $(\infty,\kappa)$-partial recursive
functions. Actually, as discussed in Section \ref{sec:generalisation}, $\Idx;$ can
actually be some \recursivelyregular ordinal $\kappa$ containing $\N$
(using Terminology from \cite{hinman2017recursion}),
in purpose to extend the sequences indexing, as proposed in
higher-order computability theories (\cite{sacks1990higher}). In that Section \ref{sec:generalisation}
case, we consider $\NN=\kappa$ starting from Section
\ref{sec:apartir}, and when we say a \finite $i$, we mean a
\metafinite{} $i$. We assume that $\NN$ comes with operations that
makes it a model of Peano axioms, and that $\Q$ is its field of
fractions. We assume that $\Idx$ can be embedded into $\NN$: That means for
example, in the case $\kappa > \omega$, that we can freely talk of a
quantity such as $\frac{1}{\omega}$. 
%
We state for now the following facts: The set of \metafinite{} $i$ is
infinite and any finite $i$ is \metafinite. 
%

}

\sabrinaOK{modifié un peu ici et motivé comme ds la reponse aux
  reviewers et déplacé début section 11}

\olivierOK{Fin new v6}

\section{\CheapM Analysis }
\label{CNSA}

\olivierOK{Rappel/MEMO  des macros:
  \begin{itemize}
  \item nx = non standard x (meme chose pour autres lettres)
  \item nxi = non standard x indice i
  \item nxn = non standard x indice n
  \item ii = indice i
  \item iin = indice n
  \item Peut etre superflus: vn = valeur n, vi = valeur i.
  \end{itemize}
}

\olivierOK{Attention a l'ordre dans lequel sont introduit les
  choses. Pour  l'instant, par exemple l'ordre est utilise avant
  d'etre defini.}

\sabrinaOK{J'ai mis des integers a la place des reels ds ttes les definitions de cette section}

We start by presenting/recalling \cheap analysis \cite{CheapNonStandardAnalysis}. 
%
It makes the distinction between two
types of mathematical objects: \textit{standard objects $\sx$} and
(strictly) \textit{non-standard objects $\nx = \nxn$}.  \Cheap objects
are allowed to depend on an asymptotic parameter $\iin \in \omega$
, contrary to standard objects that
come from classical analysis. A \cheap ``object'' $\nx$ is then
defined by a sequence $\nxn$, which is studied in the asymptotic limit
$ \iin \rightarrow \infty$, that is to say for sufficiently large
$\iin$. Every standard ``object'' $\sx$ is also considered as a
non-standard ``object'' identified by a constant sequence
$\nx=\nxn=\sx$ having value $\sx$. The underlying idea is similar to
what is for example done in probability theory where an element of
$\R$ is also implicitly considered as a probabilistic element of $\R$,
depending if it has a measure associated. This is done with ``object''
referring to any natural mathematical concept such as ``natural
number'', ``rational number'', ``real number'' ``set'', ``function''
(but this could also be ``point'', ``graph'', etc.
\cite{CheapNonStandardAnalysis}).

The idea is then to consider that all reasonings are done in the
asymptotic limit $ \iin \rightarrow \infty$, that is to say for
sufficiently large $\iin$. In particular, two \cheap elements $\nx =
\nxn$, $\ny = \nyn$ are considered as equal if $\nxn =\nyn$ after some
\finite rank.
More generally, 
any standard true relation or predicate is considered to be true for
some \cheap object if it is true for all sufficiently large values of
the asymptotic parameter, that is to
say after some \finite rank. Any
operation which can be applied to a standard object $\sx$ can also be
applied to a \cheap object $\nx=\nxn$ by applying the standard
operation componentwise, that is to say 
for each choice of the rank parameter $\iin$. 
For example, given two \cheap integers $\nx=\nxn$ and $\ny=\nyn$, we
say that $\nx>\ny$ if one has $\nxn > \nyn$ after some \finite
rank. Similarly, we say that the relation is false if it is false
after some \finite rank.
As another example, the sum $\nx+\ny$ of two \cheap integers
$\nx=\nxn$ and $\ny =\nyn$ is given by
$\nx + \ny = (\sx+\sy)_{ \iin} := \nxn + \nyn$.
A \cheap set $\nX$ is given by $\nX = \nXn$, where each $\nXn$ is a
set,  and if we write
$\nx=\nxn$, we have as expected $\nx \in
\nX$ if after some \finite rank $\nxn \in \nXn$. 
Similarly, if $\nf = \nfn: \sX_{\iin}
  \rightarrow \sY_{ \iin}$ is a \cheap function from a \cheap{} set
  $\nX$ to another \cheap set $\nY$, then $\nf(\nx)$ is the cheap
  nonstandard element defined by $\nf(\nx)=\ssf(\sx)_{\iin} :=
  \nfn(\nxn)$.
Every standard function is also a nonstandard function using all these
conventions, as expected. 
%




One key point is that introducing a dependence to a rank parameter
leads to the definition of fully new concepts: infinitely small and large
numbers. A \cheap rational $\nx=\nxn$ is infinitesimal if $0<\nx \le  \si$ for
all standard rational number $0<\si$. 
For example, $\nx=\nxn=1/\vn$ is an infinitesimal. 
A \cheap number can be infinitely large too: as an example, consider
$\comega=\comegan=\nn$ or $\nx=\nxn=2^{\nn}$, greater than any
standard number. Note that the inverse of an infinitely large \cheap
number is a \cheap infinitesimal.

  








 


From the fact that applying a standard operation to a \cheap number is
basically applying it to each possible value of $\iin$, separately,
most of the classical analysis properties on operations can hence be
transferred from the standard framework to the \cheap one: for
example, commutativity and associativity for addition and
multiplication operations.

However this is not always the case when one considers statements on
\cheap objects. One typical illustration that transfering properties
from standard predicates to \cheap ones is not automatic is the law of
excluded middle failure. Repeating \cite{CheapNonStandardAnalysis}:
For
instance, 
the nonstandard real number $\nx=\nxn
:= (-1)^{ \vn}$ is neither positive, negative, nor zero, because none of
the three statements $(-1)^{ \vn} > 0$, $(-1)^{ \vn} < 0$, or $(-1)^{ \vn} =
0$ are true for all sufficiently large $\vn$. 
%
Nevertheless, despite some peculiarities in the manipulation of
statements, 
most of the standard first-order
logic statements remains the same when quantified over \cheap
objects. 
%
We refer to \cite{CheapNonStandardAnalysis} for a very pedagogical and
more complete
discussions about \cheap concepts and some of its properties. 

\olivierOK{Etait ecrit:
We also retain a version of the equally fundamental countable
saturation property of nonstandard analysis, although the cheap
version is often referred to instead as the Arzel\`a-Ascoli
diagonalisation argument. Let us call a nonstandard property $P(\nx)$
pertaining to a nonstandard object $\nx$ satisfiable if, no matter how
one restricts the domain $\Sigma$ from its current value, one can find
an $\nx$ for which $P(\nx)$ becomes true, possibly after a further
restriction of the domain $\Sigma$. The countable saturation property
is then the assertion that if one has a countable sequence $P_1, P_2,
  \ldots$ of properties, such that $P_1(x),\ldots,P_n(x)$ is jointly
satisfiable for any finite $n$, then the entire sequence $P_1(\nx),
  P_2(\nx), \ldots$ is simultaneously satisfiable. The proof proceeds by
mimicking the diagonalisation component of the proof of the
Arzel\`a-Ascoli theorem, which explains the above terminology. 
}



\section{More on NSA: filters and ultrafilters}
\label{filter}

The classical constructions for non-standard analysis are done using
free ultrafilters.

We recall the definition of an ultrafilter over an infinite set $\I$,
called the index set. Typically, for us $\I=\omega$. 

\begin{definition}[Filter]
A filter $U$ over $\I$ is a {non-empty}
set of subsets of $\I$ such that:
\begin{enumerate}
\item $U$ is closed under superset: if $X \in U$, and $X \subset Y$,
  then $Y \in U$.
\item $U$ is closed under finite intersections: If $X \in U$ and $Y
  \in U$, then $X \cap Y \in U$.
\item $\I \in U$, but $\emptyset \not\in U$.
\end{enumerate}
\end{definition}

In particular, since $X \cap X^c=\emptyset$, and $\emptyset \not\in
U$, one cannot have both
$X$ and its complement in $U$.

\begin{lemma}[\FrechetFilter]
  The set of all cofinite (\textit{i.e.}
  complements of finite) subsets of $\I$ is a filter. It is called the \FrechetFilter.
\end{lemma}

\begin{definition}[Ultrafilter]
An ultrafilter over $\I$ is a filter $U$ over $\I$ with the additional
property that for each $X$, exactly one of the sets $X$ and $\I-X$
belongs to $U$.
A free ultrafilter is an ultrafilter $U$ such that no finite set
belongs to $U$. In the literature, a free ultrafilter is sometimes
called a non-principal ultrafilter (in opposition to principal or
fixed ones that thus contain a smallest element, called the principal
element).
\end{definition}

We just comment in the remaining lines how this relates to 
NSA. In NSA, one fixes
a free ultrafilter $U$.  One also considers sequences indexed by
$\omega$. 
Sequences $(x_i)$ and $(y_i)$ are considered equal iff the
set of indices $i$ such that $x_i=y_i$ is in the fixed free ultrafilter
$U$. Consequently, basically, \cheap analysis corresponds to the case where
$U$ is not a free ultrafilter, but the \FrechetFilter{}.
One deep
interest of the above construction (also called ultraproduct) is that
taking $U$ as a ultrafilter provides a 
transfer theorem (\Los's theorem) that guarantees that any first
order formula is true in the ultraproduct iff the set of indices $i$
when the formula is true 
belongs to the ultrafilter $U$.

%
%
\sabrinaOK{proposition d'explication:
%
}

To some extend, \cheap analysis constructions allow to reason on
objects independantly from the
ultrafilter, in the following sense (missing proofs are in appendix or
in arXiv.org version of this article\footnote{Current reference submit/2240185.}). 

\sabrinaOK{
Notice that this kind of ultraproduct construction is also present in
other areas of Mathematics such as when dealing with notions of large
numbers (such that large cardinals). Actually, this construction can
be related to the \cheap world.}

\begin{theorem}\label{th-allFreeU}
Two \cheap numbers $\na$ and $\nb$, respectively corresponding to the
sequences $(a_i)$ and $(b_i)$, are equal iff for all free ultrafilter
$U$ over $\N$ we have $(a_i)=_U (b_i)$.
\end{theorem}

\olivierOK{Etait ecrit: An more generally, some relation holds for \cheap elements
iff it holds for all free ultrafilter. pas tout a fait: on sait le dire mieux sans gros mots?
}

The following lemma is based on a statement from
\cite{keisler1976foundations}.  For selfcontentness, and completeness
we provide its proof, mostly repeating  \cite[Theorem
1.42]{keisler1976foundations} but proving also the required
extension, as we need a variation of it. 

\begin{lemma}[Folklore] \label{th:foundation}
For every infinite set $\I$, there exists a free ultrafilter over
$\I$. 
Fix some infinite set $X_0$. There exists a free ultrafilter over
$\I$ that contains $X_0$. 
\end{lemma}

\begin{proof} 
The set of all cofinite (complements of finite) subsets of $\I$ is a
filter over $\I$ (called the \FrechetFilter). 

The set of all cofinite subsets $Y$ of $\I$ and of $Y$ such that  $Y^c
\cap X_0$ is finite is also a
filter.

Let $A$ be the set of all
filters $F$ over $\I$ 
such that $F$ contains 
all cofinite subsets of
$\I$ and all $Y$ such that $Y^c \cap X_0$ is finite.

 Then $A$ is nonempty and $A$ is closed under unions of chains. 
By Zorn's Lemma, $A$ has a maximal element $U$ (in fact, infinitely
many maximal elements). 

$U$ is a filter and contains no finite set, because $U$ contains all
cofinite sets but $\emptyset \not\in  U$. 


(resp. Furthermore, for  $Y \in U$, $X_0 \cap Y$ is infinite, because
$U$ contains all $Z=Y^c$ such that 
$X_0 \cap Z^c = X_0 \cap Y$ is finite: otherwise $Y \cap Z= Y \cap Y^c=\emptyset$ but
$\emptyset \not\in  U$)

To show that $U$ an ultrafilter, we consider an arbitrary set $X
\subset \I$  and prove that there is a filter $V \supset U$  which
contains either $X$ or $\I - X$, so by maximality, $X \in U$ or $\I
- X \in U$.

Case $1$: For all $Y \in U$, $X \cap Y$ is infinite. $X$ and each $Y
\in U$ belong to the set $V =\{Z \subset \I | :Z \supset X \cap Y
\mbox{ for
some } Y \in U\}.$

$V$ is a filter over $\I$, because $V$ is obviously closed under
supersets and finite intersections, and the hypothesis of Case 1
guarantees that each $Z\in V$ is infinite.

Case 2: For some $Y \in  U$, $X \cap Y$ is finite. Then for every $W
\in  U$, $(\I-  X) \cap W$ is infinite, for otherwise $Y \cap W \in
U$  would be finite. Case 1 applies to $\I -  X$, so the set
$V =\{Z \subset \I:Z \supset (\I- X) \cap Y \mbox{for some Y } \in U\}$ is
a filter over $\I$ such that $V \subset U$, $\I- X \in V$.

We see that $X$ belongs to $U$ iff  for all $Y \in U$, $X \cap Y$ is
infinite. In particular, $X_0$ belongs to $U$ iff  for all $Y \in U$, $X_0 \cap Y$ is
infinite. Hence, $X_0 \in U$.
\end{proof}

We now go to the proof of Theorem \ref{th-allFreeU}.

\begin{proof}
Fix a free ultrafilter $U$. 
Suppose that $(a_i)$ and $(b_i)$ represent the same \cheap
number. After some \finite rank 
$a_i=b_i$. Then $\{i| a_i=b_i \}$ is in $U$ (as its
complement is finite, and hence not inside). So $(a_i)=_U (b_i)$ for that
free ultrafilter.

\olivierOK{New V6: la que super important que ensemble des meta finis est infini}
Conversely, assume that for all rank $n_0$, there is a rank $n \ge n_0$
with $a_n \neq b_n$. Then  $X=\{n| a_n \neq b_n \}$ is infinite. By
 Lemma \ref{th:foundation}, one can build a free ultrafilter $U$ with $X \in
U$. Hence $(a_i) \neq_U (b_i)$ for that
free ultrafilter $U$.
\end{proof}


\section{Computability for Integers or Rational Numbers}
\label{sec:computabilityintegers}

\subsection{Very Basic Notions From Computability}

We assume some basic familiarity with computability theory. In
computability theory, any integer $\sn \in \N$ is computable: there
exists some Turing machine $M$ that writes $\sn$ in binary.
However, not all total functions $\ssf: \N \to \N$ are computable (to
avoid ambiguities we will say \totalrecursive for ``computable'' in
this context): there does not always exist some Turing machine $M$
that takes as input $\sn$ in binary and outputs $\ssf(\sn)$ in binary.
An example of \totalrecursive function is $\sx\mapsto \sx+1$. An
example of a \totalnonrecursive function is the function $\gammaun$
which maps $\sn$ to $\sx_{\sn}+1$, where $\sx_{\sn}$ is the output of
the $n$th Turing machine on input $\sn$, for a given (non-assumed
computable) enumeration of terminating Turing machines. In what
follows $\gammaun$ will denote such a non \totalrecursive
function. 

We used the wording ``Turing machines'', but it is well known that the set of
\totalrecursive functions can be defined abstractly without referring
to Turing machines: this is the smallest set that contains the constant
function $0$, the successor function $\succ(\sx)=\sx+1$, projection
functions, and closed under composition, primitive recursion, and safe
minimization. Safe minimization is minimization over safe predicates,
that is to say predicates 
$P(\sn,\sm)$ where for all $\sn$ there is a $\sm$ with
$\sP(\sn,\sm)=1$
. 

\olivierOK{MFCS: La patch property n'est pas jolie. En fait, en gros,
  c'est le
  fait que pour une suite finie $n_1,n_2,\dots,n_k$, on peut définir
  une fonction $s(.)$ qui vaut $n_i$ en $i$, et ensuite $if(n \le k,
  s,\ssf)$ et que tout cela est bine parmi les fonctions calculables. 

ENLEVE:

The following \patchproperty{} is also true for \totalrecursive
functions: for any total function $\ssf: \N \to \N$, for any \finite
$n_0$, if there is some \totalrecursive function $\sg: \N \to \N$ such
that $\sg(n)=\ssf(n)$ for all $\sn \ge n_0$, then there is a
\totalrecursive function $\sh: \N \to \N$ such that
$\sh(\sn)=\ssf(\sn)$ for all $\sn \in \N$.
}

\sabrina{On a besoin de la patchproperty dans le Théorème \ref{besoin
    patch}. Je propose de remettre ce qui suit :}

\olivier{OK. Voir si on garde/utile}

\olivier{NOUVEAU. JE PARLE DE CALCULABILITE GENERALISE. MAIS UTILE QUE IS ON S'EN SERT.}

We will several times use the following easy remark: 
If a function is computable for all arguments above a
certain rank, then this function is computable. More formally:
\begin{theorem}[Computability for all indices]\label{pr-patchproperty}
  For any total function $\ssf: \N \to \N$, for any \finite $n_0$, if
  there is some \totalrecursive function $\sg: \N \to \N$ such that
  $\sg(n)=\ssf(n)$ for all $\sn \ge n_0$, then there is a
  \totalrecursive function $\sh: \N \to \N$ such that
  $\sh(\sn)=\ssf(\sn)$ for all $\sn \in \N$.
\end{theorem}

\sabrinaOK{probablement à rendre plus joli mais je ne vois pas comment faire ça pour l'instant}



\olivierOK{MFCS: Remarques}

\olivierOK{
COMPOSITION:
Supposons qu'on aie $\nx=\nxn=\ssf(\iin)$, et $\ny=\nyn=\sg(\iin)$

(*) SUPPOSONS que $\sg(\iin) \ge \iin$. (pour ce qui suit ne soit pas
problématique avec des décalages négatifs). 

Alors $\shiftd{\nx}{(\ny-\myomega)}=(\shiftd{\nx}{(\ny-\myomega)})_{\iin}=
\sx_{\iin + \sy_{\iin}-\myomega_{\iin}}=\ssf(\iin+\sg(\iin)-\iin)=\ssf(\sg(\iin))$

On a la composition de fonction modulo l'hypothèse (*) quand on compose.
}

\olivierOK{To Sabrina: D'accord avec ces affirmations?}

\sabrinaOK{D'accord, joli, mais ne pas oublier de définir $\myomega$}

\olivierOK{
RECURSION PRIMITIVE:  Supposons qu'on définit $\ssf$ par récursion
primitive à partir de $\sg$ et $\sh$. I.e:
$$\ssf(0,\sy)=\sg(y)$$
$$\ssf(\sx+1,\sy)=\sh(\ssf(\sx,\sy),\sx,\sy)$$

Si on s'autorise à considérer des triplets on peut considérer le triplet
$\nz=\nzn=(\ssf(\iin,\sy),\iin,\sy)$.

Il vérifie (o) $$\shiftu{\nz}=\sF(\nz)$$
si on prend $\sF(\ssf,\sx,\sy)=(\sh(\ssf,\sx,\sy),\sx+1,\sy)$ pour $\sx \neq 0$ et
$\sF(\ssf,\sx,\sy)=(\sg(\sy),1,\sy)$ pour $\sx=0$.

Si on parle de couples, on peut considérer le couple
$\nz=\nzn=(\ssf(\iin,\sy),\iin)$.

Il vérifie  (o') $$\shiftu{\nz}=\sG(\nz,\sy)$$ si on prend $\sG((\ssf,\sx),\sy)=(\sh(\ssf,\sx,\sy),\sx+1)$ pour $\sx \neq 0$ et
$\sG(\ssf,\sx)=(\sg(\sy),1)$ pour $\sx=0$.

Du coup, il a toujours parmi les solutions de $(o')$ au moins une
solution avec $\nzn=f(\iin,\sy)$.

Ca donne d'une certaine façon la fonction $\ssf$.

Donc toutes les
fonctions qu'on peut définir par récursion primitive.

Sauf, que l'on obtient la fonction qui à un indice associe $f$ pas à
partir d'un codage similaire de $g$ et $h$. 

}

\olivierOK{To Sabrina: D'accord avec ces affirmations?}

\sabrinaOK{D'accord mais $\nz=\nzn=(\ssf(\iin,\sy),\iin)$ (manque un argument ds la solution ?)}

\olivierOK{MFCS: autremetn dit, en fait, notre propriété shift c'est
  un truc qui permet essentiellement de faire des compositions de
  fonctions. Donc, avoir la stabilité par shift, ressemble à dire
  qu'on regarde des fonctions closes par composition.}

\subsection{Computable \CheapM  Numbers}
\label{sec:apartir}

\olivierOK{MFCS: Est-ce pertinant t judicieux d'écrire cela comme ca?}

In the literature, no discussions exist about the computability of
numbers: any standard number $\sn \in \A$ is computable. But here
\cheap integer or \cheap rational numbers may not be computable:

\begin{definition}[Computable \cheap number] 
A \cheap number $\nx=\nxn$ is computable if $\nxn$ seen as a sequence from $\omega$
to $\A$ is \totalrecursive.
\end{definition}

For example, $\nx=\nxn=\gammaun(\vn)$ is not a computable \cheap
integer.  Computable cheap non-standard integers include
$\comega=\comegan=\vn$.

\olivierOK{Ajout'e' car je crois qu'il est important d'expliquer le
  pourquoi du comment de cette discussion} Our purpose is first to
understand to what corresponds the subset of the computable \cheap
numbers among all \cheap numbers: can it be defined abstractly, i.e. taking \cheap analysis as a
basis (i.e. in the spirit of \cite{keislerbook} that presents mathematical
calculus taking  NSA as a basis)? 
 
\olivierOK{fin ajoute}

The following facts are easy: As usual $\ominus(x,y)$ denotes $\max(0,x-y)$.

\begin{theorem}[Stability by \totalrecursive functions]
For any \cheap computable numbers $\nx_1$, $\nx_2$, \dots, $\nx_k$, for
any standard \totalrecursive function $\ssf: \A^k \to \A$, we have that 
$\ssf(\nx_1,\nx_2,\dots,\nx_k)$ is a
computable \cheap number.
\end{theorem}

\olivierOK{New V6: Attention: pas toujours vrai cas ordinal. Pour raison
  evidente}

\begin{theorem}[Basic properties]
The set of \cheap computable natural numbers is a semiring: In
particular, it is stable by $+$, $\ominus$,
$\cdot$. 
The set of \cheap computable rational numbers is a ring: In
particular, it is stable by $+$, $-$, inverse, $\cdot$. 
\end{theorem}

\begin{theorem}[First characterization]\label{th:tr}
  \begin{itemize}
  \item The set of \cheap computable numbers is the smallest set that
    contains $\comega$ and that is stable by standard \totalrecursive
    $\ssf: \NN \to \A$.
\item  The set of \cheap computable numbers is also the set of
  $\ssf(\comega)$ for \totalrecursive standard $\ssf: \NN \to \A$.
  \end{itemize}
\end{theorem}
\begin{proof}
  The \Cheap integer $\comega$ is computable. When $\nx$ is computable
  and $\ssf$ is a standard \totalrecursive function, then, $\ssf(\nx)$
  is computable. Now, from definitions $\nx=\nxn$ is computable, iff
  $\nxn=\ssf(\iin)$ for some standard \totalrecursive $\ssf$, hence
  $\nx=\ssf(\comega)$. First item follows.

  Second item is a direct corollary of above reasoning. 
\end{proof}

\subsection{\Shift Operation and Preservation Property}

\olivierOK{MFCS: RAPPEL de note importante: en fait, notre propriété shift c'est
  un truc qui permet essentiellement de faire des compositions de
  fonctions. Donc, avoir la stabilité par shift, ressemble à dire
  qu'on regarde des fonctions closes par composition.}

The previous properties can also be stated in another alternative way:
Consider the following operation $\shift$ that maps \cheap numbers to \cheap
numbers: 

\begin{definition}[Shift operation]
  Whenever $\nx=\nxn$, $\shiftu{\nx}$ is defined by
  $\shiftu{\nx}=(\shiftu{\nx})_{\iin}=\sx_{\iin+1}$.
\end{definition}

Notice that 
$\shiftu{\sn}=\sn$ for all standard integer $\sn$.
However, $\shiftu{\nx}$ is not necessarily $\nx$ for a \cheap
$\nx$. In particular, $\shiftu{\comega}=\comega+1$. In other words,
using a non-standard analysis inspired vocabulary, $\shift$
is not an internal operation. 

\olivierOK{Added V6}
\begin{theorem}\label{besoin patch} \label{th:avecpatch}
The set of computable \cheap numbers is the smallest set
that contains all solutions of $\shiftu{\nx}=\ssf(\nx)$ for $\ssf$ standard total
recursive,
and that is stable by standard \totalrecursive $\ssf: \NN
\to \A$.
\end{theorem}

\olivierOK{MFCS: est ce qu'on veut parler d'équations plus générale du
  type (o') et de primitif récursif...}

\begin{proof}
  \Cheap $\comega$ integer can be obtained as a solution of
  $\shiftu{\comega}=\comega+1$. Hence this class contains all
  computable \cheap numbers from above statements.

  We only need to state that \cheap numbers are stable by such a
  $\shift$ equation: Assume that $\nx=\nxn$ is solution of
  $\shiftu{\nx}=\ssf(\nx)$. Then after some \finite rank $\iin_0$, we
  must have $\sx_{\iin_0+\ik}=\ssf^{[\ik]}(\sx_{\iin_0})$, where
  $\ssf^{[\ik]}$ denotes $\ik$th iteration of $\ssf$ (computability of
  $\sx_{\iin_0}$ follows from Theorem \ref{pr-patchproperty}). This
  yields computability for indices $\iin \ge \iin_0$. And hence, this
  yields computability for all $\iin$ by Theorem
  \ref{pr-patchproperty}.
\end{proof}




A key remark is that the unary $\shift$ operation can actually be
extended to a binary operation. A \cheap element of $\omega$ is called a
\cheap index.

\begin{definition}[$Shift$]
Given some \cheap number $\nx=\nxn$ and some \cheap index
$\ny=\nyn$, let $\shiftd{\nx}{\ny}$ be defined by $\shiftd{\nx}{\ny}= (\shiftd{\nx}{\ny})_{\iin}=
\sx_{\iin + \sy_{\iin}}.$
\end{definition}

It can be checked that this is a valid definition: its value is
independant of the representative. It can also be checked that it
satisfies $\shiftd{\nx}{0}=\nx$, $\shiftd{\nx}{1}=\shiftu{\nx}$,
$\shiftd{\nx}{(\ny+1)}=\shiftu{(\shiftd{\nx}{\ny})}$,
$\shiftd{\nx}{(\ny+\nz)}=\shiftd{(\shiftd{\nx}{\ny})}{\nz}$ for any
\cheap number $\nx$ and cheap indices $\ny$ and $\nz$.

\olivierOK{On est bien d'accord que ca marche toujours,
  meme si ordinaux !!!}


\begin{theorem}
Assume that $\nx$ and $\ny$ are computable. Then
$\shiftd{\nx}{\ny}$ is computable. 
\end{theorem}

From previous definitions, we derive easily the following preservation
property.

\begin{theorem}[Preservation property] \label{pr-transfer}
Let $\sP$ be some standard property over the numbers. If 
$\sP(\nn_1,\dots,\nn_k)$ holds for non standard numbers
$\nn_1,\dots,\nn_k$, then $\sP(\shiftu{\nn_1},\dots,\shiftu{\nn_k})$
holds.

More generally, 
$\sP(\shiftd{\nn_1}{\nn},\dots,\shiftd{\nn_k}{\nn})$ holds 
for all \cheap index $\nn$. 
\end{theorem}

In some axiomatic view, computability of \cheap numbers can be
summarized as follows:

\begin{theorem}
\begin{itemize}
\item Not all \cheap numbers are computable.  
\item Computable \cheap numbers include all standard numbers. The
  image of a computable \cheap number by a standard \totalrecursive
  function is computable.
\item The infinitely large \cheap number 
  $\comega$ satisfying $\comega=\comegan=\iin$ is among computable
  numbers. 
\item Computable \cheap numbers are exactly those that can be obtained by above rules. 
\item There exists some operation $\shiftu{(.)}$ over \cheap numbers,
  that preserves standard numbers, and that satisfies preservation
  property (Theorem \ref{pr-transfer}).


\end{itemize}
\end{theorem}


\sabrina{
subsection{\Cheap Minimization}

\olivier{MFCS: Le grand karma tel qu'ecrit dans LICS n'est pas très
  joli. Je tente des reformulations plus jolies. }

\olivier{ Il faut Fixer une fois pour
  toute l'ordre des arguments n et m. Disons qu'on fait P(n,m) et
  qu'on minimise sur m}

\olivier{MFCS: Pour m'aider à comprendre. MAIS Attention c'est pas le Grand karma (contrairement à ce
  que j'ai pensé à un moment), formulation sans cheap non-standard: Given a \totalrecursive predicate
$\sP(\ov n,m)$, It is said to be safe if for all $\ov \sn$, there exists
  some $\sm$ with $ \sP(\ov \sn,\sm)=1$.

  In that case, $\mu \sm~ \sP(\ov \sn,\sm )$ is computable (total
  recursive).

Preuve: C'est un résultat classisque voir une définition selon comment on
définit les fonctions calculables. 
(source pour la terminologie ``safe'': cf ce qui est commenté dans le .tex à cet
endroit précis.
}

\olivier{En fait, le fond du problème des discussions à venir reliée à la possible indécidabilité
  est que juste qu'étant donnée $\nx$, on 
  $\sP(\nx)$ n'est pas forcément décidable.

Autrement dit, si $\nx=\nxn=\ssf(\nx)$, écrire $\sP(\nx)$ est une
propritété qui quantifie sur les arguments de $f$ quand $\nx$ n'est
pas standard.
}

\olivier{MFCS: La meme chose, mais formulation avec cheap
  non-standard. (rq: ca parait plus magique).

(version prédicat unaire). 
Consider a unary \totalrecursive predicate
$\sP(.)$. Assume there is $\nm$ such that $\sP(\nm)=1$. Then there is a
least one $\mu \nm$.

This least one $\mu \nm$ is computable. 

TO Sabrina: on est d'accord?
}

\olivier{MFCS: 
Grand Karma (version prédicat plus général). 
Consider a \totalrecursive predicate
$\sP(.)$. Assume there exist some computable $\nx_1,\nx_2,\dots,\nx_k$ 
and some (not necessarily computable) \cheap integer $\nm$ such that $\sP(\nx_1,\nx_2,\dots,\nx_k,\nm)=1$. Then there is some 
least one  $\mu \nm$ such that $\sP(\nx_1,\nx_2,\dots,\nx_k,\nm)=1$.

$\mu \nm$ is computable iff it is bounded by some computable $\ny$.

}

\olivier{MFCS (Une application qu'on a en tete,  pour comprendre): 

En fait, on voudrait utiliser cela pour ce type d'arguments (le texte
est du copié collé de trucs aui apparaissent après)

Let $\ny$ some \cheap number. Let $0<\nepsilon$ be any computable
infinitesimal.

Asume both ar computable. 

We know that here always exists some \cheap \finite index $\nn$ with
$0< \shiftd{\nepsilon}{\nn} < \ny.$

Consider predicate $P(\nepsilon,\ny,\nn)$ given by ``$0<
\shiftd{\nepsilon}{\nn} < \ny.$''

On voudrait alors dire.

Then $\mu~ \nn~ P(\nepsilon,\ny,\nn)$ is computable iff it is bounded
by some computable $\ny$. 

C'est (plus ou moins vrai ) pour exactement les memes arguments que ce
qui suit,  

Mais SAUF QUE CA DECOULE PAS DU GRAND KARMA ACTUAL: $P$ n'est 
n'est pas un prédicat standard. On parle dans
ce prédicat de 
``$\shiftd{\nepsilon}{\nn}$'' et pas de $\nepsilon$ et $\nn$. Le
prédicat ici ``joue'' avec les indices... via des shifts...
}

A key argument in computability is the following: given a
\totalrecursive predicate $\sP(\sm)$, if we know that $\sP(\sm)=1$
holds for some $\sm$, then such an $\sm$ can be found, and actually
even the least such $\sm$, usually denoted by $\mu m~
P(m)$. Basically, test for $\sk=0,1, \dots$ until $\sP(\sk)=1$. This
will terminate, as it will not go further that $\sk=\sm$.

An interesting discussion is the case where $\nm$ is not a standard
\integer. \olivier{Added: pour aider à comprendre} In that case, one
cannot directly apply the above strategy, since first there is no
simple way to enumerate all \cheap integers, and second, given $\nk$,
determining whether $\sP(\nk)=1$ is not always possible effectively.

However, the following holds: 

\olivier{Formulation pour plus que prédicat binaire?}

\olivier{Tentative de formulation juste cf email.}

Assume that $\sP(\nn,\nm)$ holds for some \cheap \integer{} $\nm$. Let
us denote $\mu \nm$ (or sometimes also $\mu \nm~\sP(\nn,\nm)$) the
least such \cheap \integer{}: $\sP(\nn,\mu \nm)=1$, and
$\sP(\nn,\nx)=1$ implies $\mu \nm \le \nx$. We describe a minimization
algorithm to compute in a effective way $\mu \nm$:

Assume that $\mu \nm \le \ny$ for some computable $\ny$. There must
exists some \finite rank $\iin_0$ such that for $\iin \ge \iin_0$,
$(\mu \nm)_{\iin} \le \nyn$.  As $\ny$ is computable, we have
$\ny=\nyn=\sg(\iin)$ for some \totalrecursive $\sg$.  Then $(\mu
\nm)_{\iin}$ can be computed for $\iin \ge \iin_0$ by testing from
$\sk=0,1,\dots, \sg(\iin)$ whether $\sP(\sh(\iin),\sk)=1$ where
$\nz=\nzn=\sh(\iin)$. This procedure will terminate for $\iin \ge
\iin_0$ before $\sk=\sg(\iin)$ as we know that $(\mu \nm)_{\iin} \le
\sg(\iin)$. It follows that $\mu \nm$ is computable, as the \finitely
many other values of $(\mu \nm)_{\iin}$ for $\iin \le \iin_0$ can be
hard coded in a program.

From this algorithm:

\begin{theorem}[Computability of minimization] \label{grandtout} Let
  $\sP(\sn,\sm)$ be a \totalrecursive standard predicate where $\sn$
  and $\sm$ are some \integers. Let $\nn$ is some computable \cheap
  number. Then $\mu \nm$ exists and is computable.

 Conversely,  if 
$\mu \nm= \mu \nm~ \sP(\sn,\nm)$ is computable, then it
is true that $\mu \nm \le \ny$ for some \cheap computable $\ny$

\end{theorem}

\sabrina{Et ds ce cas le deuxième cas de la preuve serait quelque chose comme : par définition, la séquence des $\sP(\sn,\nm)$ est totale récursive et on peut tester qui est le plus petit car il existe un rang fini à partir duquel le plus petit est trouvé ?.}
\olivier{Bon, la preuve est aussi à changer en fonction de l'énoncé}
\begin{proof}
  The fact that there exists some least such \cheap \integer{} follows
  from the fact that there exists some least integer for each element
  of the sequence. The computability is given by the previous algorithm.

\sabrina{ du coup pas sûre pour cette ancienne suite de l'autre sens :
Assume that $\mu \nm$ is not computable. Write $\mu \nm = (\mu \nm)_{\iin} =
\smu(\iin)$ for some  (non computable) function $\smu(\iin)$. 

Let $\sX \subset \Idx$ be some infinite set. Typically $\sX=\Idx$, but it
could be another subset of $\Idx$.  Assume that $\smu(\iin)$ is not
computable on $\sX$ in the following sense: there is no algorithm that would return $\smu(\iin)$
for all $\iin \in \sX$.

That means that
 for any \totalrecursive function $\psi(\iin)$ there cannot exists a
 \finite rank
 $\iin_0$ such that for all  $\iin \ge \iin_0, \iin \in \sX$, $\smu (\iin) \le \psi(\iin)$, as
 otherwise we could use the argument above to get computability of
 $\smu(\iin)$ over $\sX$ using bound $\psi(\iin)$.  In other words,
 that means that for all \finite $\iin_0$ there exists
 some $\iin \ge \iin_0, \iin \in \sX$ with $\smu (\iin) > \psi(\iin)$. This can be restated as
 $C_\psi=\{\iin | \smu (\iin) > \psi(\iin)\}$ is such that $C_\psi \cap X$ is
 infinite. 
\olivierOK{New V6: la que super important que ensemble des meta finis est infini}

Given some other \totalrecursive function $\psi'(\iin)$, reasoning with
$\sX:=C_\psi \cap \sX$, we must have $C_\psi \cap C_{\psi'} \cap \sX$ is
infinite, otherwise, we would get computability of $\smu(\iin)$ over $\sX$
by considering bound $\max(\smu(\iin),\smu'(\iin))$.

Let $\psi_1,\psi_2, \dots, \psi_i, \dots$  be some (not necessarily
recursive in any sense) enumeration of \totalrecursive functions.

Starting from $\sX=\Idx$, and repeating the argument, for all $\sk$,
$\bigcap_{i \le \sk} C_{\psi_i}$ is infinite. For all $k \in \N$, let  $\lambda(k)$ be any
\integer{}  in that set, chosen bigger than $\lambda(k')$ for $k' \le k$. 

For any cheap non standard computable $\ny$, as $\ny=\nyn=\psi(\iin)$
for some \totalrecursive $\psi$, we cannot have $\mu \nm \le \ny$ as
this increasing sequence $\lambda(\iin)$ is providing an infinite set of ranks $\iin$ on which the
property $(\mu \nm )_{\iin} \le y_{\iin}=\psi(\iin)$  is not true
after some rank. }
\end{proof}



Notice that we prove that $\mu \nx= \mu \nx~ \sP(\nx)$ is not
computable implies that this is not true that $\mu \nx \le \ny$ for
some \cheap computable $\ny$. But it does not mean that $\mu \nx >
\ny$ for all \cheap computable $\ny$: for classical non-standard
analysis $<$ would be a total order, but this is not the case for
cheap non-standard analysis.

As \totalrecursive functions can be enumerated (we do not say
recursively enumerated, in any sense), \cheap computable integers can be
enumerated.

\begin{definition}[Enumeration of computable cheap
non-standard integers] 
We denote $\nn_1,\nn_2, \dots, \nn_i, \dots $ an enumeration of computable cheap
non-standard integers. 
\end{definition}

Hence we get the following key remark: under the hypotheses of the
above proposition: $\mu \nx$ is computable iff $\mu \nx \le \nn_i$ for
some $i$.

\olivierOK{MFCS: Etait commente pour gagner de la place, mais remis car
  vraiment nécessaire pour comprendre le GRAND KARMA.}

If one is able to decide whether $\sP(\nn_i)=1$ for a given $i$ (not
always possible in classical computability), one can execute the
following kind of algorithm to compute $\mu \nx$: test for
$i=0,1,\dots$ whether $\sP(\nn_i)=1$. As soon as such an $i$ is found,
return $\mu \nm$ using the minimization algorithm described above
(using bound provided by $\nn_i$). From above statements, this will
return $\mu \nm$, and terminates if (and actually only if) $\mu \nm$
is computable.

}
\section{Infinitesimals and Infinitely Large Numbers}
\label{sec:infinitesimals}

%
Any \cheap rational number $\nx$ is of the form $\np/\nq$ for some \cheap
\integers{} $\np$ and $\nq$. It is computable iff it is of the form $\np/\nq$ with  $\np$ and $\nq$
computable. 

\Cheap \integers{} as well as \cheap rational numbers can be
infinitely large. \Cheap rational numbers can also be infinitesimals.
For example, $\frac{1}{\comega+1}$ and $2^{-\comega}$ are computable
infinitesimals. \Cheap rationals $\nx=\nxn=\frac{1}{\gammaun(\vn)}$ as
well as $\nx=\nxn=2^{-\gammaun(\vn)}$ are non-computable
infinitesimals. 
\olivier{Enlevé, car superflus un peu: In all
what follows, for shortness reasons, when talking about an
infinitesimal number, it will always be assumed to be some (\cheapP)
rational number.}

\begin{definition}[Infinitely large and infinitesimal numbers]
A \cheap number $\nx$ is infinitely large iff for all \emph{standard} number $\sy$,
one has $\nx \ge \sy$.
A \cheap rational number $0<\nx$ is infinitesimal iff for all
\emph{standard} rational number $0 < \sy$,
one has $0< \nx \le \sy$.
\end{definition}




One key point in the above concept is that this involves a
quantification over all \emph{standard} number $\sy$, which is weaker
than over all \cheap $\ny$.
Actually, we however have the following phenomenon:

\begin{theorem} \label{pr-th:truc}
  Let $\nx$ (respectively: $0< \nx $) be some \cheap number that is
  infinitely large (resp. infinitesimal).
For any \cheap number $\ny$ (resp. $0 < \ny$)  there exists some
\cheap number $\nx'$, of the form $\nx'=\shiftd{\nx}{\nn}$ for some \cheap
\finite index $\nn$, 
with $\nx' \ge \ny$ (resp. $0< \nx' \le  \ny$).
\end{theorem}



\begin{proof}
Let $\ny$ some \cheap number. Write $\ny=\nyi$. 

Consider some $\ii \in \N$.  As $\nx=\nxn$ is infinitely large
(respectively: infinitesimal),
there must exists some \finite rank $\iin_0$ such that for all $\iin \ge
\iin_0$, we have $\nxn \ge \nyi$ (resp. $0< \nxn \le \nyi$). 

Let $\sg: \N \to \N$ be the function
that maps $\si$ to the corresponding $\iin_0$ for all $\si \in
\N$. Consider \cheap index $\nn$ defined by
$\nn=\nn_{\ii}=\sg(\ii)$. 

From definitions $\nx'=\shiftd{\nx}{\nn}$ is 
such that $\nx'=\nx'_{\ii}=\nx_{\ii+\sg(\ii)}$ and hence satisfy $\nx'_{\ii} \ge
\nyi$ (resp. $0 < \nx'_{\ii} \le \nyi$)
for all $\ii$.  The conclusion follows.
\end{proof}


\olivier{MFCS: comme dit plus haut, ce qui suit n'est pas une
  application direte du grand KARMA.

Ca parle plutot du prédicat:
$P(\nepsilon,\ny,\nn)$ given by ``$0<
\shiftd{\nepsilon}{\nn} < \ny.$''
que d'un prédicat standard $\sP$ appliqué à $(\nepsilon,\ny,\nn)$
i.e. de $\sP(\nepsilon,\ny,\nn)$. (on utilise un décalage....)
}

\olivier{MFCS: strict ou pas strict. Fixer une fois pour tout. Disons
  pas strict}


Fix some computable infinitesimal $0<\ny$. We have that for any computable infinitesimal 
$0<\nepsilon$, there always exists some \cheap \finite index $\nn$ with
$0< \shiftd{\nepsilon}{\nn} \le \ny.$ However, this $\nn$ can be
non-computable. 
%
Therefore, it is natural to consider the following notion.












\begin{definition}[Effectiveness]
We say that some computable infinitesimal $0<\nepsilon$ is
\machintrucby computable $0<\ny$ iff 
 there exists
some \cheap {computable} index $\nn$ with
$\shiftd{\nepsilon}{\nn} \le  \ny.$
\end{definition}

This is clearly a reflexive relation as $\shiftd{\nepsilon}{0}=\nepsilon$. It is also transitive:

\begin{theorem}[Transitivity of computably bounded relation]\label{transitivity}

Let $\nepsilon, \nepsilon', \ny$ be some non zero positive computable
infinitesimals. If $\nepsilon$ is \machintrucby $\nepsilon'$ and
$\nepsilon'$\machintrucby $\ny$, then $\nepsilon$ is \machintrucby
$\ny$.
\end{theorem}

\begin{proof}
 If $\nepsilon$ is \machintrucby $\nepsilon'$ and
  $\nepsilon'$ \machintrucby{}  $\ny$ then there exists some \cheap
  {computable} \finite index $\nn_i$ with $\shiftd{\nepsilon}{\nn_i} \le
  \nepsilon'$, and some \cheap {computable} \finite index $\nn_j$ with
  $\shiftd{\nepsilon'}{\nn_j} \le \ny$. But then
  $\shiftd{\nepsilon}{(\nn_i+\nn_j)} \le \ny$: Indeed, apply
  $\shiftd{(.)}{\nn_j}$ to members of $\shiftd{\nepsilon}{\nn_i} \le
  \nepsilon'$ to get
  $$\shiftd{\nepsilon}{(\nn_i+\nn_j)}=\shiftd{(\shiftd{\nepsilon}{\nn_i})}{\nn_j}
  \le \shiftd{\nepsilon'}{\nn_j} \le \ny$$ from previously stated
  properties of $\shift$ operation. 
\end{proof}

As a consequence, the following notion is natural and provides an
equivalence relation:

\begin{definition}
We say that two computable infinitesimals $0 < \nepsilon$ and $0 <
\nepsilon'$ are \machinequivalent iff $0 < \nepsilon$ is \machintrucby  $0 < \nepsilon'$ and conversely.
\end{definition}


\begin{theorem}\label{compEpsilon}
  $\frac{1}{\comega+1}$ is \machintrucby any computable infinitesimal
  $0 < \nepsilon$.
\end{theorem}

\begin{proof}
Consider computable \cheap index $\nm$ given by $\nm=\sg(\nepsilon)$ where
standard function $\sg(\iin)=\partientieresup{1/\iin} - 1$ for $\iin \ge \iin_0$, and
say $\sg(0)=1$ (as $0<\nepsilon$, its components are non-zero after
some rank). Then $\shiftd{\left(\frac{1}{\comega+1}\right)}{\nm} < \nepsilon$. 

\olivier{MFCS: C'est une nouvelle preuve plus dans l'esprit non-standard.

L'ancienne preuve était:
Consider $\nepsilon=\nepsilon_n=\ssf(\iin)$ be some computable
infinitesimal, i.e. with $\ssf(\iin)$ \totalrecursive. As $0 <
\nepsilon$, there exists some \finite index $\iin_0$ such that $f(\iin)\neq 0$ for
$\iin \ge \iin_0$.  Consider computable \cheap index $\nm=\nmn=\sg(\iin)$ where
$\sg(\iin)=\partientieresup{1/\ssf(\iin)} - 1$ for $\iin \ge \iin_0$, and
say $\sg(\iin)=1$ for $\iin < \iin_0$.  Then
$\shiftd{\frac{1}{\comega+1}}{\nm}=\left(\shiftd{\frac{1}{\comega+1}}{\nm}\right)_{\iin}=\frac{1}{\iin+1
  + \sg(\iin)} <  
\ssf(\iin)$ for all $\iin \ge \iin_0$
and hence $\shiftd{\left\(\frac{1}{\comega+1}\right)}{\nm} < \nepsilon$. 
}
\end{proof}

A computable
infinitesimal $0 < \nepsilon$ is said to be monotone if
$\nepsilon=\nepsilonn$ with $\sepsilon_{\iin+1} \le \nepsilonn$ for all
$\iin$. 
%
Monotone computable infinitesimals include $\frac{1}{\comega+1}$ and
$2^{-\comega}$.

\begin{theorem}\label{monotone}
    Any monotone computable infinitesimal $0 < \nepsilon$ is
    \machintrucby $\frac{1}{\comega+1}$.
    All monotone computable infinitesimals are \machinequivalentn.
\end{theorem}

\olivier{MFCS: Formulation de la chose qui suit sans infinitesimaux:  Consider $\nepsilon=\nepsilon_n=\ssf(\iin)$ be some computable
  monotone infinitesimal, i.e. with $\ssf(\iin)$ \totalrecursive and
  decreasing.  


 As $\nepsilon$ is
  infinitesimal, given any standard integer $\sn$, there exists
   $\sn_0=\spsi(\sn)$ such that for all $\iin \ge \spsi(\sn),
  \ssf(\iin) \le \frac{1}{\sn+1}$.

 Since $\ssf(\spsi(\sn)) \le \frac{1}{\sn+1}$, we get that predicate
 $\sP(\sn,\sm)$ given by $f(\sm) \le \frac{1}{\sn+1}$ is safe. 
 It follows that $\mu \sm~ \sP(\sn,\sm)$ is computable. Consider
 $\nm =\nmn=\mu \sm~ \sP(\sn,\sm)$. We have
$\shiftd{\nepsilon}{\nm} \le \frac{1}{\comega+1}$.
}

\olivier{Assez agacant: dans ce qui suit encore autre chose que le grand karma en
  fait. Ca parle du prédicat $P(\nepsilon,\comega,\nm)$ = ``$\shiftd{\nepsilon}{\nm} \le
\frac{1}{\comega+1}$'' qui n'est pas standard.

Ce qui est sympa avec ce prédicat, c'est que c'est un prédicat safe:
si on se donne un rang $n$, on cherche $f(m) \le \frac{1}{\sn+1}$ et
on sait que ca existe car $f$ est décroissant et de limite $0$. 
}

\begin{proof}
  Consider $\nepsilon=\nepsilon_n=\ssf(\iin)$ be some computable
  monotone infinitesimal, i.e. with $\ssf(\iin)$ \totalrecursive and
  decreasing. 
From Theorem \ref{pr-th:truc}, there exists for some \cheap
\finite index $\nm$, with $\shiftd{\nepsilon}{\nm} \le
\frac{1}{\comega+1}$.  We get that predicate
 $\sP(\sn,\sm)$ given by $f(\sn +\sm) \le \frac{1}{\sn+1}$ is safe. 
 It follows that $\mu \sm~ \sP(\sn,\sm)$ is computable. Consider
 $\nm' =\nm'_{\iin}=\mu \sm~ \sP(\sn,\sm)$, hence computable. We have
$\shiftd{\nepsilon}{\nm'} \le \frac{1}{\comega+1}$.

\olivier{Old preuve: pas si jolie
Replacing if needed $\ssf$ by $1$ for the
  \finitely many values on which it is $0$ in the following reasoning
  (using patchproperty),
  we may assume that $\ssf(\iin) \ge 1$ for all $\iin$. 
 Consider predicate $\sP(\sn,\sm)$ defined by $
 \ssf(\sm) < \frac{1}{\sn+1}$.  Apply the reasoning of proposition
\ref{grandtout} to predicate $\sP$, observing that $\sP(\comega,\comega)=1$ as $\ssf$ is monotone
and unbounded by any \integer{} $1/(\sn+1)$ for any standard \integer{}  $\sn$,
and observing that $\comega$ is computable. 
Hence, there 
must exists some computable \cheap \integer{}  $\nn=\mu \nn$ with
$\sP(\comega,\nn)=1$. Then $\shiftd{\nepsilon}{\nn} \le \frac{1}{\comega+1}$.

}
First statement follows.

Second statement is a clear corollary.
\end{proof}






We say that some computable infinitesimal $0<\nepsilon$ is \ieffective
if it belongs to the above class: it is monotone or computably
equivalent to some monotone computable infinitesimal.


\begin{corollary} 
A computable infinitesimal $0<\nepsilon$ is \ieffective iff it is
\machintrucby  $\frac{1}{\comega+1}$. 
%
Any  \ieffective computable infinitesimal is \machintrucby any computable
$0 < \ny$. 
\end{corollary}

\section{Computability for Real Numbers}
\label{sec:reals}

Functions from the reals to the reals are the main studied functions
in computable analysis. That's why after having studied computability
for \cheap integer and rational numbers, we now go to computability
for \cheap real numbers.

\begin{definition}[Computability for real numbers]
Fix some \ieffective  computable infinitesimal $\nepsilon$.
A standard real $\sx$ is said to be computable if there exists some
\cheap computable rational $\frac{\np}{\nq}$, such that
$\left|\sx-\frac{\np}{\nq}\right| \le \nepsilon.$
\end{definition}

\begin{theorem} \label{machin}
The previous definition is not depending on $\nepsilon$:  if this holds
for an \ieffective infinitesimal $0<\nepsilon$, then it holds for
any other \ieffective computable infinitesimal $0<\nepsilon'$. 
\end{theorem}





\begin{proof}

Infinitesimal $\nepsilon$ is \machintrucby  $1/\nnu$  for
$\nnu=2^{\partientieresup{\log (2/\nepsilon')}}$. Let
$\Snepsilon=\shiftd{\nepsilon}{\nn}$ the corresponding infinitesimal:  $\Snepsilon \le 1/\nnu$, that is $\nnu  \Snepsilon < 1$.
  Let $\Snp=\shiftd{\np}{\nn}$ and
 $\Snq=\shiftd{\nq}{\nn}$. We know (using preservation Theorem \ref{pr-transfer}) that $\left|\sx-\frac{\Snp}{\Snq}\right| \le
\Snepsilon.$
%
%
Consider then 
$\np'=\partientieresup{\nnu \Snp/ \Snq}$, $\nq'=\nnu$. 
This guarantees $\left|\sx-\frac{\np'}{\nq'}\right| \le
\nepsilon'.$
Indeed, 
$|\nnu \sx - \nnu \frac{\Snp}{ \Snq} |
= |\nnu| |\sx-\frac{\Snp}{ \Snq}| \le \nnu \Snepsilon < 1,$
implies $|\nnu \sx - \partientieresup{\nnu \frac{\Snp}{ \Snq}} | \le
|\partientieresup{\nnu \frac{\Snp}{ \Snq}}-\nnu \frac{\Snp}{ \Snq}|+|\nnu \sx- \nnu \frac{\Snp}{ \Snq} 
| \le 1 + 1 = 2$
using definition of what integer part is,
and  then  $|\partientieresup{\nnu \frac{\Snp}{ \Snq}}/ \nnu -\sx| \le
2/\nnu = 2^{1-\partientieresup{\log (2/\nepsilon')}} \le \nepsilon'.$

\end{proof}

  Two \cheap reals are said to be infinitely close (respectively:
  effectively infinitely close) if the absolute value of their
  difference is less than some (resp. \ieffectiven) computable
  infinitesimal $0<\nepsilon$.
%

\begin{definition}[Left and Right-Computability for real numbers]
A standard real $\sx$ is said to be left-computable (respectively:
right-computable)  if it is infinitely close to some \cheap computable
rational $\frac{\np}{\nq}$ with $\frac{\np}{\nq}\le \sx$
(resp. $\frac{\np}{\nq}\ge \sx$).
\end{definition}

\begin{theorem}\label{stdcomp}
    A standard real $\sx$ is computable iff it is effectively
    infinitely close to some \cheap computable rational     $\frac{\np}{\nq}$.  
A standard real $\sx$ is computable iff it is
    right-computable and left-computable.
\end{theorem}

\begin{proof}
First statement is just a restatement of the definition. Concerning
second statement: 
Direction from left to right of second item is trivial.  Direction
from right to left is the following.\olivier{MFCS: Needs a proof with \cheap
  reasoning}
Assume that $\sx$ is right and left-computable. 
There exists some $\np=\npn=\ssp(\iin)$, $\nq=\nqn=\sq(\iin)$,
$\np'=\np'_{\iin}=\ssp'(\iin)$, $\nq'=\nq'_{\iin}=\sq'(\iin)$ such
that $\frac{\np}{\nq}\le \sx \le \frac{\np'}{\nq'}$ and
$\sx-\frac{\np'}{\nq'}$ and $\sx-\frac{\np}{\nq}$ both
infinitesimal. This must hold componentwise for $\iin \ge \iin_0$ for
some $\iin_0$. Replacing if needed the values of functions for indices
less than $\iin_0$, we can assume without loss of generality that $\iin_0=1$. 

Given $\iin$, consider $\phi(\iin)=\max_{1 \le \ii \le \iin}
\frac{\ssp(\ii)}{\sq(\ii)}$, and $\phi'(\iin)=\min_{1 \le \ii \le \iin}
\frac{\ssp'(\ii)}{\sq'(\ii)}$.  $\phi$ (respectively: $\phi'$) is an
increasing (resp. decreasing) 
function converging to $\sx$. We have
$$\phi(\iin) \le \sx
\le \phi'(\iin).$$
The  predicate $\sP(\sn,\sm)$ given by $\phi'(\sm)-\phi(\sm) \le
\frac{1}{\sn}$ is safe.
Consequently,  $\mu \sm \sP(\sn,\sm)$  is computable, and 
$\phi(\mu
\sm \sP(\sn,\sm))$ is a computable sequence of rational numbers
proving that $\sx$ is computable, since considering $\frac{\np}{\nq}=\left(\frac{\np}{\nq}\right)_{\iin}=\phi(\mu
\im \sP(\iin,\im))$, 
we get $$\left|\frac{\np}{\nq}-x\right| \le \frac{1}{\comega}.$$ 
\end{proof}

Let $\sD=\{\sr \in \Q | \sr=\frac{\sn}{2^{\sm}} \mbox{ for \integers{}
} \sn,\sm\}$: these are the rationals with finite binary
representation. They are sometimes also called dyadic rationals.

\begin{theorem}\label{d}
We can always assume $\frac{\np}{\nq} \in \nD$ in previous statements,
i.e. $\nq$ to be of the form $2^{\nm}$ for some \cheap \integer{} $\nm$. 
\end{theorem}

\begin{proof}
This is the case in the proof of Theorem \ref{machin}. The case of
left and right-computability is similar. 
\end{proof}

One important theorem is that this corresponds to the classical
definition of computability for reals (in the sense of computable
analysis):   Formally,  according to classical definitions and statements from
\cite{Wei00,Ko91,PourEl:comap},  this is equivalent to say that the following holds:

\begin{theorem}\label{th:un}
  A standard real $\sx$ is computable iff there exist some
  \totalrecursive functions $\ssp(\sn)$ and $\sq(\sn)>0$ such that
  $|\sx- \frac{\ssp(\sn)}{\sq(\sn)}| \le \frac{1}{2^{\sn}}$ for all
  integer $\sn$.
  A standard real $\sx$ is left-computable (resp. right-computable)
  iff there exist some \totalrecursive functions $\ssp(\sn)$ and
  $\sq(\sn)>0$ such that $\sx=\sup_{n} \frac{\ssp(\sn)}{\sq(\sn)}$
  (resp. $\sx=\inf_{n} \frac{\ssp(\sn)}{\sq(\sn)}$).
\end{theorem}

\begin{proof}
  Consider monotone computable infinitesimal
  $\nepsilon=\nepsilonn=\frac{1}{2^{\iin}}$. There must exist some
  computable \cheap integers $\np$ and $\nq$ such that
  $|\frac{\np}{\nq}-\sx| \le \nepsilon$. The \totalrecursive functions
  $\ssp$ and $\sq$ such that $\np=\npn=\ssp(\iin)$ and
  $\nq=\nqn=\sq(\iin)$ satisfies the above property after some \finite
  rank $\iin_0$. They can be fixed to $\ssp(\iin_0)$ and $\sq(\iin_0)$
  on the finitely many $\iin$ before $\iin_0$ so that this holds for
  all $\iin$.

Conversely, if this holds, $\nepsilon=\nepsilonn=\frac{1}{2^{\iin}}$
is a monotone computable infinitesimal, and $\np$ and $\nq$ such that
$\np=\npn=\ssp(\iin)$ and $\nq=\nqn=\sq(\iin)$ are computable \cheap integers such that
$|\frac{\np}{\nq}-\sx| \le \nepsilon$.

The statements for right and left-computability are obtained in a
similar fashion. 
\olivierOK{MFCS: besoin d'une preuve, ou on dit similaire, ou c'est clair.}

\end{proof}

\section{Continuity and Effective Uniform Continuity for Real
  Functions}
\label{sec:continuity}

The following theorem is left as an exercice in \cite{CheapNonStandardAnalysis}. 

\olivier{En fait, c'est pour $\ssf: \sX \rightarrow {\R}$ be a
  standard function on a standard metric space $\sX = (\sX,d)$ et nous
  on simplifie.

Par ailleurs, Terence Tao dit aussi: (Subsequential nonstandard
  version) If $\sx$ is a standard element of $\sX$ and $\ny$ is a cheap
  nonstandard element that is infinitesimally close to $\sx$, then
  $\ssf(\ny)$ is infinitesimally close to $\ssf(\sx)$ after passing to a
  subsequence.
}

\olivierOK{MFCS question de fond: 
Formulation alternative:  EST-CE MIEUX?
A function $\ssf: \sX \subset \R \to \R$ is continuous iff for all
standard element $\sx$ of $\sX$, and for all infinitesimal $\ndelta$
we have that  $\ssf(\sx + \ndelta)$ is
  infinitesimally close to $\ssf(\sx)$.  
}

\begin{theorem} \label{th:continuous}
A function $\ssf: \sX \subset \R \to \R$ is continuous iff 
for all
standard element $\sx$ of $\sX$, and for all \cheap element $\ny$
infinitely close to $\sx$, then $\ssf(\ny)$ is
  infinitesimally close to $\ssf(\sx)$.  
\end{theorem}

We provide here the proof for completeness: 


\begin{proof}
Function 
$\ssf$ is continuous in $\sx$ iff for all $\sepsilon$ there
exists some $\sdelta$ such that whenever $|\sx-\sy| \le
\sdelta$ we have $|\ssf(\sx)-\ssf(\sy)| \le \sepsilon$.

For the right to the left direction: Assume $\ssf$ is continuous. Consider
some standard $\sx$ of $\sX$, and some \cheap element $\ny$ infinitely
close to $\sx$. Consider some standard $0 < \sepsilon$, and the
corresponding  standard $\sdelta$. As $\sx-\ny$ is infinitesimal,
writing $\ny=\nyn$, 
there is some $\iin_0$ such that for all $\iin \ge \iin_0$, we have
$|\sx-\nyn| \le \sdelta$. Consequently, $|\ssf(\sx)-\ssf(\nyn)| \le
\sepsilon$, that is to say we have $|\ssf(\sx)-\ssf(\ny)| \le
\sepsilon$. As this holds for all standard $\sepsilon$, $\ssf(\ny)$ is
infinitely close to $\ssf(\sx)$.

For the left to the right direction: Assume $\ssf$ is not continuous. That means that there exists some
$\sepsilon$ such that for all $\sdelta$, say $\sdelta(\sn)=\frac{1}{\sn}$, there exists some
$\sy(\sn)$ with $|\sx-\sy(\sn)| \le
\sdelta(\sn)$ and $|\ssf(\sx)-\ssf(\sy(\sn))| > \sepsilon$.
That means that $\ny=\nyn=\sy(\sn)$ is some \cheap element infinitely
close to $\sx$ but with $\ssf(\ny)$ not infinitely close to
$\ssf(\sx)$.


\end{proof}

\olivierOK{MFCS: question de fond. Est-ce mieux de le formuler en:
A function $\ssf: \sX \subset \R \to \R$ is uniformly continuous  iff for all
\cheap element $\nx$ of $\sX$, and for all infinitesimal $\ndelta$
we have that  $\ssf(\nx + \ndelta)$ is
  infinitesimally close to $\ssf(\sx)$.  }

Similarly, the following can be established: 

\begin{theorem}\label{th:ucontinuous}
  A function $\ssf: \sX \subset \R \to \R$ is uniformly continuous if
  for all \cheap element $\nx$ of $\sX$, and for all \cheap element
  $\ny$ infinitely close to $\nx$, then $\ssf(\ny)$ is infinitesimally
  close to $\ssf(\nx)$.
\end{theorem}

\olivier{MFCS: la remarque qui suit ca vient de wikipedia. Lui rendre
  hommage. }

\sabrinaOK{Observe the difference between definitions of continuity and uniform
continuity where quantification is over any standard or all \cheap
$x$. For functions over some compact domain, both notions are
equivalent, but for general functions, a function can be continuous
without being uniformly continuous.} For example standard function $x
\mapsto x^2$ with domain $\R$ is
not-uniformly continuous as $(\nx+1/\nx)^2= \nx^2 + 2+ \frac{1}{\nx^2}$ is not
infinitely close to $\nx^2$ when $\nx$ is infinitely large. However,
it
is uniformly-continuous (and hence continuous) on $[0,1]$ as for $\ny$ infinitesimal, 
$(\nx+\ny)^2= \nx^2 + 2 \nx \ny + \ny^2$ is always infinitely close to
$\nx^2$ when $\nx \in [0,1]$ (hence is bounded). 



Notice that above theorems are defining concepts of continuity and
uniform continuity very elegantly: there is no alternance of
quantifiers compared to the classical $\epsilon$-$\delta$
definition. 
Refer to
\cite{NSAteaching} for practical measurements on the benefits of
NSA  concepts in teaching.

\section{Computability For Real Functions}
\label{sec:functions}

We now go to computability issues: 

\begin{definition}
Fix some \ieffective computable infinitesimal 
$0<\nepsilon$.
%
A function $\ssf: \sX \subset \R \to \R$ has an effective modulus of continuity iff there exists some computable
nonstandard $\ndelta$ such that for all standard $\sx$ and $\sy$, 
$
\mbox{if } |\sx-\sy| \le \ndelta \mbox{ then }
|\ssf(\sx)-\ssf(\sy)| \le 
\nepsilon.$
\end{definition}

Obviously, such a function is uniformly continuous, and hence continuous.
More fundamentally:

\begin{theorem} \label{thduex}
 The previous concept is not depending on $0<\nepsilon$: if this
  holds for a \ieffective computable infinitesimal $0<\nepsilon$, then
  it holds for any other \ieffective computable infinitesimal
  $0<\nepsilon'$.
\end{theorem}

\begin{proof}
Infinitesimal $\nepsilon$ is \machintrucby  $\nepsilon'$. Let
$\Snepsilon=\shiftd{\nepsilon}{\nn}$ the corresponding infinitesimal:
that is to say 
$\Snepsilon \le \nepsilon'$. But then 
$\Sndelta=\shiftd{\nepsilon}{\nn}$ provides the property for
$\nepsilon'$ as we know using Theorem \ref{pr-transfer} that
  $$
\mbox{if } |\sx-\sy| \le \Sndelta \mbox{ then }
|\ssf(\sx)-\ssf(\sy)| \le 
\Snepsilon \le \nepsilon'.$$
\end{proof}

\begin{theorem} \label{thstepun}
We can always assume $\ndelta \in \nD$ in previous statements,
i.e. $\ndelta$ to be of the form $\frac{\np}{2^{\nm}}$ for some \cheap \integer{} $\nm$. 
\end{theorem}

\begin{proof}
This follows from the proof of previous theorem as applying the sense from left
to right, and then left to right provides a $\ndelta$ of  that form. \olivier{Une preuve directe?}
\end{proof}

This still corresponds to the classical notion from computable
analysis. Formally,  according to classical definitions and statements from
\cite{Wei00,Ko91,PourEl:comap},  this is equivalent to say that the following holds:

\begin{theorem} \label{th-th:nd}
A function $\ssf: \sX \subset \R \to \R$ has an effective modulus of
continuity iff there exists some \totalrecursive
$\sm: \N \to \N$ such that if $|\sx-\sy| \le 
  2^{-\sm(\sn)}$ then $|\ssf(\sx)-\ssf(\sy)| \le  2^{-\sn}$ for all
  standard $\sx,\sy \in \sX$. 
\end{theorem}

\begin{proof}
Assume there is such a recursive $\sm$. Consider monotone computable \cheap infinitesimal 
 $0<\nepsilon$ given by $\nepsilon=\nepsilonn=2^{-\iin}$.  Consider
 computable \cheap rational $\ndelta$ given by 
$\ndelta=\ndeltan=2^{-\sm(\sn)}$. Then $$\mbox{if } |\sx-\sy| \le \ndelta \mbox{ then }
|\ssf(\sx)-\ssf(\sy)| \le 
\nepsilon$$
holds. 

Conversely, assume that function $\ssf$ has an effective modulus of continuity. Consider monotone computable
infinitesimal $\nepsilon=2^{-\comega}$. There must exists some
computable \cheap  $\ndelta$ such that
$$\mbox{if } |\sx-\sy| \le \ndelta \mbox{ then } |\ssf(\sx)-\ssf(\sy)|
\le \nepsilon.$$

That means that there exists some \finite rank $\iin_0$ such that the
properties holds componentwise for $\iin \ge \iin_0$. 
Write $\ndelta=\ndeltan=\ssf(\iin)$ for some \totalrecursive
$\ssf$. Consider $\sm(\iin)=\partientieresup{-\log(\ssf(\iin))}$ for $\iin \ge \iin_0$ and
$\sm(\iin)=\sm(\iin_0)$ for $\iin < \iin_0$. This provides the
expected property for all $\iin$. 
\end{proof}

\olivier{Commente:
There are several equivalent ways to define the concept of computable
function $f:[a,b]\rightarrow\mathbb{R}$ in computable analysis. See
e.g. \cite{Wei00,Ko91,PourEl:comap}.  We follow in this section inspiration from
\cite{Ko91}.
}

\olivierOK{meilleurs refs que 
\url{https://en.wikipedia.org/wiki/Computable_real_function}.
}

Consider an indexed family of \cheap numbers: to  some parameter
$\si$, is associated some \cheap number $\nx(\si)=\nx(\si)_{\iin}$.
We say that the family is uniformly computable in $\si$ if there exists
some standard \totalrecursive function $\ssf(\si,\iin)$ such that $\nx(\si)_{\iin}=
\ssf(\si,\iin)$.

\begin{definition}[Computability for functions over the reals]
 Fix some \ieffective computable infinitesimal $\nepsilon$. We say that $\ssf:
  [\sa,\ssb] \subset \R \rightarrow {\R}$ (standard domain) is
  computable iff
\begin{enumerate}

\item{} [discretization property]: 
there exists  some computable $\ndelta$ such that 
$$
\mbox{if } |\sx-\sy| \le \ndelta \mbox{ then }
|\ssf(\sx)-\ssf(\sy)| \le 
\nepsilon.$$
\item {} [it has some uniform approximation function over the
  rationals]:
There exists some indexed family of \cheap rationals
$\npsi(\sq)$, uniformly computable in $\sq $, 
such that 
$$|\npsi(\sq)- \ssf(\sq)| \leq \nepsilon$$
for all $\sq \in \Q \cap  [a,b]$

\olivier{Bref on a changé la définiton: 

Ancienne définiton était: there exists some \totalrecursive $\spsi :( \D \cap [\sa, \ssb]) \times \N
                                \to \D$ 
such that for all standard rational $\sd \in \D$, standard integer $\sn$
  $$ |\spsi(\sd, \sn) - \ssf(\sd)| \leq 2^{-n}.$$
}

\olivier{A ete a un moment:
there
  exists some \totalrecursive $\spsi: \Q \times \QP \to \Q$ such that $$
  |\spsi(\nq,\nepsilon) - \ssf(\nq)| \leq \nepsilon$$
for all \cheap rational number $\nq \in [a,b]$
\olivier{MFCS: introduire notation $\QP$}

C'est un intermédiaire, mais moins joli probablement. 
}

\end{enumerate}
\end{definition}

Before going to the statement and proof that this corresponds to the
classical notion of computability for functions over the reals, notice
that one main interest of the above definition is that it sounds morecl
natural and easier to grasp than classical ones\footnote{See statement
  of Theorem \ref{th-th-deft} for example of a classical definition.}:
in particular, item $1$. is a very natural discretization
property\footnote{It follows from the proofs that the
    discretization property is equivalent to the existence of an
    effective modulus of continuity. However, we believe the latter
    concept is harder to grasp, as basically talking about the
    dependence of a $\delta$ from $\epsilon$, whose meanings is not so
    natural.}.


\begin{theorem} \label{pr-th:inde}
The previous definition is not depending on $\nepsilon$.
\end{theorem}

\begin{proof}
This follows from Theorems \ref{thduex} for item
1. Item 2. holds for some effective $\epsilon$ iff it holds for any
effective $\epsilon$ using a reasoning similar to Theorem
\ref{thduex}. \olivier{Is that clear? (and true)}
\end{proof}

\begin{theorem} \label{pr-thdyadic}
Without loss of generality, we can always assume $\ndelta$ and $\sq$ in above definiton to be in
$\nD$ and $\D$, i.e. to be of
the form $2^{\nm}$ for some \cheap \integer{} $\nm$ or $\sm$ instead of being
rational numbers. 
\end{theorem}

\begin{proof}
The fact that this is true for item 1. is Theorem \ref{thstepun}.
Now, if 1. holds, then $\psi(\nq)$ can be replaced by $\psi(\nq')$
where $\nq' \in \nD$ is approximating \cheap rational $\nq$ at
precision $\ndelta$. The error would then be at most $2\nepsilon$
instead of $\nepsilon$. But as this is equivalent to hold for
$\nepsilon$ by above Theorem, we get the statement.
\olivier{Is that clear? (and true)}
\end{proof}

It  turns out that our definition is equivalent to the classical notion of
computability in computable analysis.    Formally,  according to classical definitions and statements from
\cite{Wei00,Ko91,PourEl:comap},  this is equivalent to say that the following holds:

\begin{theorem} \label{th-th-deft}
\label{Teo:Ko}A real function $\ssf:[\sa,\ssb]\rightarrow\mathbb{R}$ is
computable iff it is computable in the sense of computable analysis. 
\end{theorem}

\begin{proof}
It is proved for example in \cite[Corollary
2.14]{Ko91}  that $\ssf$ as above is computable in
the sense of computable analysis iff 
\begin{enumerate}
\item{}  [it has an effective modulus of continuity] 
there exists some \totalrecursive $\sm: \N \to \N$ such that if $|\sx-\sy| <
  2^{-\sm(\sn)}$ then $|\ssf(\sx)-\ssf(\sy)| < 2^{-\sn}$ for all
  standard $\sx,\sy \in \sX$. 
\item{}  [it has some computable approximation function]: there exists some \totalrecursive $\spsi : \D \cap [a,b] \times \N
                                \to \D$ 
such that for all standard rational $\sd \in \D$, standard integer $\sn$
  $$ |\spsi(\sd, \sn) - \ssf(\sd)| \leq 2^{-\sn}.$$
\end{enumerate}

Now,  Item 1. is equivalent to our Item 1. by Theorem
\ref{th-th:nd}. Concerning Item 2. Using Theorem
\ref{pr-thdyadic},  and Theorem \ref{pr-th:inde}, considering infinitesimal
$2^{-\comega}$, suppose  there
  exists some $\npsi'$ such that $$
  |\npsi'(\sd) - \ssf(\sd)| \leq \nepsilon=2^{-\comega}$$
for all \cheap number $\sd \in \D \cap [a,b]$. 
Write $\npsi'(\sd)=\npsi'(\sd)_{\iin}=\spsi(\sd,\iin)$. Above inequality  yields item 2 above.

Conversely, assume we have $ |\spsi'(\sd, \sn) - \ssf(\sd)| \leq
2^{-\sn}$ for some $\spsi'$. Consider
$\npsi(\sd)=\npsi(\sd)_{\iin}=\spsi'(\sd,\iin)$. This yields item 2. of our
definition. 
\olivier{On est d'accord? }

\olivier{il faut mieux parle calculabilité unforme ou pas? Ou ce qu'on
  dit suffit.}
\end{proof}

\sabrinaOK{colorer en orange les \cheap numbers dans la preuve}

\sabrinaOK{Regrouper les sections 7 et 8 dans une meme section sur les fonctions de
  reels et faire deux sous-sections pour respectivement la continuite
  et la calculabilite ?}

\section{Examples of Applications}
\label{sec:example}

Hence, as expected, results known about computable functions are true
in this framework, and conversely.
However, our framework can present alternative ways to establish
proofs.

\olivierOK{MFCS: Important. Dans le cas de l'analyse non-standard
  non-cheap, pour $x$ limité, on peut parler de $st(\nx)$, l'unique réel tel que
  $\nx-st(\nx)$ soit infinitésimal.

Il existe car on regarde $X=\{t \in \R| t < \nx\}$ qui est un ensemble
non-vide et majoré, dont on peut prendre le sup. On vérifie alors que
$\sup(X)$ vérifie la définition de $st(\nx)$. 

Mais dans le cas cheap, on ne peut pas parler de $X$. L'ordre $<$
n'est pas bien ordonné.

A vrai dire, une raison profonde c'est que pour l'entier \cheap
$(-1)^{\iin}$ il n'a pas de partie standard.
}

\olivier{MFCS NEW:}

Notice that \cheap analysis is however distinct from NSA, and some of
NSA statements and concepts needs to be adapted. As an example,
a \cheap number $\nx=\nxn$ is said to be limited if there exists some
standard real $\sy$ such that $|\nx| \le \sy$. In NSA,  to every limited non-standard
number $\nx$ is associated some unique standard real number $st(\nx)$,
called its standard-part, such that $\nx-st(\nx)$ is infinitesimal. 
This is not possible in \cheap analysis since for example
$\nx=\nxn=(-1)^{\iin}$ is clearly non infinitely close to any standard
real number.
We can however talk about the following:

\olivier{MFCS: Supprimer celui dont pas besoin.}
\begin{definition}[Standard part $st^+(\nx)$ and $st^-(\nx)$]
Assume $\nx=\nxn$ is limited.
We write $st^-(\nx)$ (respectively: $st^+(\nx)$) for the limit inf
(resp. limit sup) of $\iin \mapsto \nxn$.
\end{definition}

We write $inf(\nx)$ (respectively: $sup(\nx)$) for the inf
(resp. sup) of $\iin \mapsto \nxn$.
The following is easy to establish: \olivier{On est d'accord?}
\begin{lemma}
Assume $\nx=\nxn=\ssf(\iin)$ is limited, where $\ssf$ is some standard
function.

Some standard $\sy$ is some accumulation point of $\iin \mapsto
\ssf(\iin)$ iff there exists some infinitely large \cheap index $\nN=\nNn$,
monotone (i.e. $\nN_{\iin+1}\ge \nNn$), such that $|\ssf(\nN)-\sy|$ is
infinitesimal. 
Consequently, 
$st^-(\nx)$ and $st^+(\nx)$ are respectively the least and largest
such $\sy$, i.e. infinitely close to some $\ssf(\nN^-)$ and $\ssf(\nN^+)$.
\olivier{smaller devrait etre least. on dit least et quoi en anglais en fait?}
\end{lemma}

The following two statements have very elegant classical  proofs in
NSA: See e.g. \cite{robinson1996non,keislerbook} This can be adapted to a proof using \cheap
arguments. 

\olivierOK{a,b std : couleur:}

\begin{theorem}[Intermediate Value Theorem] \label{th-th:arnaque}
Every continuous standard function $\ssf:[\sa,\ssb] \to \R$ with $f(\sa)  \cdot f(\ssb) <0$
has a zero at some standard $\sx$. 
\end{theorem}

\olivierOK{Il faut écrire la preuve}

\label{w-th:arnaque}

We need the following whose proof is easy.

\begin{lemma} \label{corobesoin}
Assume that $\ssf: \R \to \R$ is some continuous function. Assume that
$\nx$ is limited. If $\ssf(\nx)>0$ then $\ssf(st^-(\nx)) \ge 0$.
If  $\ssf(\nx)<0$ then $\ssf(st^-(\nx)) \le 0$.
Similarly for $st^+$. 
If $\ssf(\nx)>0$ then $\ssf(inf(\nx)) \ge 0$.
If  $\ssf(\nx)<0$ then $\ssf(sup(\nx)) \le 0$.
Similarly for $st^+$. 

\end{lemma}

We now prove Theorem \ref{th-th:arnaque}.

\begin{proof}
  Assume w.l.o.g that $[\sa,\ssb]=[0,1]$ and that $\ssf(0)<0$ and $\ssf(1)>0$.  Take infinitely large
  \cheap integer $\nN$.  The idea is to consider the $\nx(\nk)$ of the form $\nk \cdot
  \frac{1}{\nN}$ for \cheap integer $0 \le \nk \le \nN$.


  To do so, consider $\nk^-=\min(\nS^+)$ where $\nS^+=\nS^+_{\iin}$
and $$\nS^+_{\iin}=\{0 \le \sk \le \nNn \mbox{ and } \ssf(\sk \cdot
  \frac{1}{\nNn})  \ge 0\}.$$ 
\olivier{Oui: $\nNn=\iin$ et donc ca s'écrit plus simplement} 

\olivier{Question: Ca peut aussi s'écrire du coup:
$$\nk^-=\min(\{0 \le \nk \le \nN \mbox{ and } \ssf(\nk \cdot
  \frac{1}{\nN}) \ge 0\})$$
Est-ce plus clair ou mieux?
En tout les cas, plus subtile car le $\ssf(\nk \cdot
  \frac{1}{\nN}) \ge 0$ est en fait en train de parler composante par composante
}

This latter set is non-empty as $\ssf(1)>0$ and does not contain $0$
since $\ssf(0)<0$. 

  Function $\ssf$ is continuous, hence uniformly continuous on its
  domain. 
Since $\nx(\nk^-)$ and $\nx(\nk^-) - \frac{1}{\nN}$ are infinitely
close, necessarily $\ssf(\nx(\nk^-))$ and $\ssf(\nx(\nk^-) -
\frac{1}{\nN})$ must be infinitely close by Theorem
  \ref{th:ucontinuous}. 


We have $\ssf(\nx(\nk^-)) \ge 0$ and $\ssf(\nx(\nk^-) - \frac{1}{\nN}) <
0$ by definition of $\nk^-$.
Consequently, using Lemma \ref{corobesoin}, necessarily $\ssf(st^+(\nx(\nk^-))) \ge 0$ and
$\ssf(st^+(\nx(\nk^-)))=\ssf(st^+(\nx(\nk^-) - \frac{1}{\nN})) \le
0$, hence $\ssf(\sx)=0$ for standard $\sx=st^+(\nx(\nk^-))$.

\olivier{MFCS: Ok pour cette preuve?}
\end{proof}

Notice that we could have considered $st^-$, or the $\max(\nS^-)$
defined symmetrically, and this could provide possibly other zeros. 

\begin{theorem}[Extreme Value Theorem] \label{th:extreme}
Every continuous  standard function $\ssf:[\sa,\ssb] \to \R$ attains
its maximum at some standard $\sx$. 
\end{theorem}

We need the following whose proof is easy.

\begin{lemma} \label{lelemmequimanque}
Assume that $\ssf: \R \to \R$ is some continuous function. Assume that
$\nx$ is limited and $\sm$ is some standard value.
If $\ssf(\nx) \le \sm$ then $\ssf(st^-(\nx)) \le \sm$.
\end{lemma}

We can now go to the proof of Theorem \ref{th:extreme}:

\begin{proof}
  Assume w.l.o.g that $[\sa,\ssb]=[0,1]$. Take infinitely large
  \cheap integer $\nN$.  The idea is to consider the $\nx(\nk)$ of the form $\nk \cdot
  \frac{1}{\nN}$ for \cheap integer $0 \le \nk \le \nN$.


  To do so, consider $\nk^-=\min(\nS)$, 
where $\nS=\nS_{\iin}$
and $$\nS_{\iin}=\{0 \le \sk \le \nNn \mbox{ and } \ssf(\sk \cdot
  \frac{1}{\nNn})  \ge \ssf(\sk' \cdot
  \frac{1}{\nNn}) \mbox{ for all } 0 \le \sk' \le \nNn
\}.$$ 
\olivier{Oui: $\nNn=\iin$ et donc ca s'écrit plus simplement} 

\olivier{Question: Ca peut aussi s'écrire du coup
  autrement. L'écrire. 
}

This latter set is non-empty as a finite set always has a maximum. 

  Function $\ssf$ is continuous, hence uniformly continuous on its
  domain.  

Consider $\sx=st^-(\frac{\nk^-}{\nN})$, and $\sm=f(\sx)$. Then we claim that $f(\sy) \le \sm$
for all standard $\sy$. Indeed, any $\sy$ contain at least one $\nk' \cdot
  \frac{1}{\nN}$, $0 \le \nk' \le \nN$ infinitely close to it: 
  Consider $\nk'=\lceil \nN
  \cdot \sy \rceil$. 

Hence $\ssf(\sy)$ is infinitely close to $\ssf(\nk' \cdot
  \frac{1}{\nN})$. The latter, is less than $\ssf(\nk^- \cdot
  \frac{1}{\nN})$ by definition of $\nk^-$, and hence less than $m$ by
  Lemma \ref{lelemmequimanque}.
\end{proof}

Notice that we could have considered $st^+$, or the $\max(\nS^-)$
defined symmetrically, and this could provide possibly other maximums.

We now adapt these proofs to go to computability issues:


\olivier{To Sabrina. Ca serait bien de rédiger des preuves
  alternatives pour les mettre en annexes pour comparaison. Je veux
  dire rédiger les preuves dans l'esprit de celels qui étaient au
  tableau (et dans tes notes) du théorème des valeurs
  itermédiaires. En essayant, par rapport à ce qui était au tableau,
  d'uniformiser les notations pour aider à la comparaison.}


\olivier{MFCS: first without computability, as this is not directly
  the same as classical non-standard analysis. Furthermore, it may
  help.}

\begin{lemma}
  \label{corobesoind}
  Assume that $\nx$ is limited and computable. Then $inf(\nx)$ is
  right-computable and $sup(\nx)$ is left-computable.
\end{lemma}

\olivier{On est d'accord?}




\olivierOK{a,b couleur. Verifier partout}

\begin{theorem} \label{th:arnaquec}
Every computable function $\ssf:[\sa,\ssb] \to \R$ with $f(\sa)  \cdot f(\ssb) <0$
has a right-computable zero and a left-computable zero.
\end{theorem}

\begin{proof}
 Assume w.l.o.g that $[\sa,\ssb]=[0,1]$ and that $\ssf(0)<0$ and
 $\ssf(1)>0$. 

Consider effective infinitesimal $\nepsilon=\nepsilonn$. There must
exists some computable $\ndelta$ such that $ \mbox{if } |\sx-\sy| \le
\ndelta \mbox{ then } |\ssf(\sx)-\ssf(\sy)| \le \nepsilon.$ There must
also exists some indexed family of \cheap rationals
$\npsi(\sq)=\npsi(\sq)_{\iin}=\spsi(\sq,\iin)$, uniformly computable
in $\sq $, such that $|\npsi(\sq)- \ssf(\sq)| \leq \nepsilon$ for all
$\sq \in \Q \cap [0,1]$.

Consider infinitely large computable \cheap integer $\nN=\max(\comega,\frac{1}{\ndelta})$.  The idea is to consider the $\nx(\nk)$ of the form $\nk \cdot
  \frac{1}{\nN}$ for \cheap integer $0 \le \nk \le \nN$.

  To do so, consider $\nk^-=\min(\nS^+)$ where $\nS^+=\nS^+_{\iin}$
and $$\nS^+_{\iin}=\{0 \le \sk \le \nNn \mbox{ and } \spsi(\sk \cdot
  \frac{1}{\nNn},\iin) \ge -\nepsilonn\}.$$ 

\olivier{Question: Ca peut aussi s'écrire du coup:
$$\nk^-=\min(\{0 \le \nk \le \nN \mbox{ and } \spsi(\nk \cdot
  \frac{1}{\nN},\comega) \ge -\nepsilon\})$$
Est-ce plus clair ou mieux?
En tout les cas, plus subtile car le $\spsi(\nk \cdot
  \frac{1}{\nN},\comega) \ge -\nepsilon$ est en fait en train de parler composante par composante
}


  Function $\ssf$ is continuous, hence uniformly continuous on its
  domain. 
Since $\nx(\nk^-)$ and $\nx(\nk^-) - \frac{1}{\nN}$ are infinitely
close, necessarily $\ssf(\nx(\nk^-))$ and $\ssf(\nx(\nk^-) -
\frac{1}{\nN})$ must be infinitely close by Theorem
  \ref{th:ucontinuous}.

We have $\ssf(\nx(\nk^-)) \ge - 2\epsilon $ and $\ssf(\nx(\nk^-) - \frac{1}{\nN}) <
0$ by definition of $\nk^-$ and $\npsi$.

\olivier{On est d'accord sur ces majorations ?}


Consequently\footnote{Formally, we are using implicitely Lemma
  \ref{corobesoin}, which is easy to establish from definitions.} necessarily $\ssf(st^+(\nx(\nk^-))) \ge 0$ and
$\ssf(st^+(\nx(\nk^-)))=\ssf(st^+(\nx(\nk^-) - \frac{1}{\nN})) \le
0$, hence $\ssf(\sx)=0$ for standard $\sx=st^+(\nx(\nk^-))$. 

Furthermore by the
 property about $\ssf$ and $\ndelta$ and $\nepsilon$, we are sure that 
 $\ssf(\sy)<0$ for all $\sy < \sx$. Consequently, we also have
 $\sx=\sup(\nx(\nk^-))$. \olivier{On est bien d'accord sur ce
   paragraphe au dessus?}
From Lemma \ref{corobesoind}, it is
right-computable. 

Considering $st^+$, and the $\max(\nS^-)$
defined symmetrically, provides  a left-computable
zero.

\olivierOK{
ANCIENNE PREUVE, MOCHE. 
  Without loss of generality we
  can assume $\ssf(0)<0$ and $\ssf(1)>0$. 
  We split the discussion into two cases.  First case: There is a standard
  rational number $\sx$ with $\ssf(\sx) = 0$. Then $\ssf$ has a
  computable zero, namely  $\sx$. 

  Second case: For all standard rational number $\sx$ we have
  $\ssf(\sx) \neq 0$. As $\ssf$ is a continuous function, we know it
  has a zero, and actually that $\sx^*=\inf\{\sx|\ssf(\sx)> 0\}$ is a
  particular zero. 

  Split interval $[0,1]$ into
  $\nn=2^{\partientieresup{\log (\comega)}}$ subintervals: Namely,
  consider $\nx_{\nk}:=\frac{\nk}{\nn}$ for $\nk$ a \cheap \integer{}
  with $1 \le \nk \le \nn$. 

Observe that given some $\nk$, as we know that $\ssf(\nx_{\nk}) \neq
0$, either $\psi(\nx_{\nk}, \frac{1}{\nm})<-\frac{1}{\nm}$ for some
$\nm$ in which case we can deduce that $\ssf(\nx_{\nk}) <
0$, or $\psi(\nx_{\nk}, \frac{1}{\nm})>\frac{1}{\nm}$ for some
$\nm$ in which case we can deduce that $\ssf(\nx_{\nk}) >
0$. This provides computability of predicate $\ssf(\nx_{\nk}) >0$
(respectively $\ssf(\nx_{\nk}) <0$).


%
%
  Consider the least $\nk$ such that $\ssf(\nx_{\nk})<0$ and
  $\ssf(\nx_{\nk+1})>0$. Necessarily, 
  $\sx^* \in (\nx_{\nk}, \nx_{\nk+1})$. We have
  $|\nx_{\nk}-\sx| \le \frac{1}{\nn}$ which is an \ieffective
  infinitesimal. 
  Now, $\nx_{\nk}$ is a computable \cheap rational number
  from its definition and using above remarks.



}
\end{proof}

\begin{corollary} \label{coro:avantrice}
Every computable function $\ssf:[\sa,\ssb] \to \R$ with $f(\sa)  \cdot f(\ssb)
<0$ has a standard zero $\sx$.
It this zero is isolated (there exists some standard $\sepsilon>0$ such that
$\ssf$ has no other zero on $[\sx-\sepsilon,\sx+\sepsilon]$), then it is computable.
\end{corollary}
\begin{proof}
The existence of $\sx$ follows from previous theorem (Intermediate
Value Theorem). If $\sx$ is isolated, then by considering $\ssf$ on
$[\sx-\sepsilon,\sx+\sepsilon]$ in previous Theorem, we get that this (unique) zero $\sx$ is
left-computable and right-computable. Hence, it is computable.
\end{proof}

All this can be used to prove for example Rice's theorem. 


\begin{theorem}[Rice's theorem] \label{th:rice}
The set of standard computable reals is a real closed field.
\end{theorem}
\begin{proof} 
  Standard computable reals are closed by addition, subtraction,
  multiplication and division. We do the proof for multiplication,
  other proofs are similar. Fix some effective infinitesimal
  $0<\nepsilon$.  Assume $\sx$ and $\sy$ are computable.  Let $\sK$ be
  some standard constant such that $|\sx| \le \sK$ and $|\sy| \le
  \sK$.  Consider effective infinitesimal $\nepsilon'=
  \frac{\nepsilon}{2 \sK+1}$.  

  There must exist some \cheap computable rationals $\frac{\np}{\nq}$
  and $\frac{\np'}{\nq'}$ such that $\left|\sx-\frac{\np}{\nq}\right|
  \le \nepsilon'$ and $\left|\sy-\frac{\np'}{\nq'}\right| \le
  \nepsilon'.$ Then
\begin{eqnarray*}
\left| \sx \cdot \sy - \frac{\np
  \cdot 
  \np'}{\nq \cdot \nq'} \right| &\le&  \left| \sx -
\frac{\np}{\nq}\right| \cdot |\sy| + \left| \frac{\np}{\nq} 
\right| \cdot \left| y - \frac{\np'}{\nq'} \right| \\
& \le& \sK \nepsilon' + (\sK+\nepsilon') \nepsilon' \\
&=& (2\sK+1) \nepsilon' \\&=& \nepsilon,
\end{eqnarray*}
bounding the $\nepsilon'$ in term  $\sK+\nepsilon'$ by $1$. 

Similarly, it is easy to establish that polynomials with coefficients that are
  standard computable reals are computable. Then given such a
  polynomial, if it has a real root $\sx$, then one can always find
  some standard rational $\sa,\ssb$ such that $\sx$ is the only root
  in interval $[\sa,\ssb]$. One can then apply previous theorem
  (Intermediate Value Theorem) on the
  polynomial restricted to this interval to get that it must have a
  computable root. This computable root can only be $\sx$.
\end{proof}

With the same principle, the following can be established: 

\begin{theorem} \label{th:extremec}
Every computable function $\ssf:[\sa,\ssb] \to \R$ attains its maximum
in a right-computable standard point and in a left-computable standard
point. If a maximum point is isolated, then it is computable.
\end{theorem}

The proof is similar to Theorem \ref{th:arnaquec}, but adapting the
proof from Theorem \ref{th:extreme}. 

\begin{proof}
 Assume w.l.o.g that $[\sa,\ssb]=[0,1]$. 

Consider effective infinitesimal $\nepsilon=\nepsilonn$. There must
exists some computable $\ndelta$ such that $ \mbox{if } |\sx-\sy| \le
\ndelta \mbox{ then } |\ssf(\sx)-\ssf(\sy)| \le \nepsilon.$ There must
also exists some indexed family of \cheap rationals
$\npsi(\sq)=\npsi(\sq)_{\iin}=\spsi(\sq,\iin)$, uniformly computable
in $\sq $, such that $|\npsi(\sq)- \ssf(\sq)| \leq \nepsilon$ for all
$\sq \in \Q \cap [0,1]$.

Consider infinitely large computable \cheap integer $\nN=\max(\comega,\frac{1}{\ndelta})$.  The idea is to consider the $\nx(\nk)$ of the form $\nk \cdot
  \frac{1}{\nN}$ for \cheap integer $0 \le \nk \le \nN$.

  To do so, consider $\nk^-=\min(\nS)$ where $\nS=\nS_{\iin}$
and
$$\nS_{\iin}=\{0 \le \sk \le \nNn \mbox{ and } \spsi(\sk \cdot
  \frac{1}{\nNn},\iin)  \ge \spsi(\sk' \cdot
  \frac{1}{\nNn},\iin) - \nepsilonn \mbox{ for all } 0 \le \sk' \le \nN
\}.$$ 
\olivier{Oui: $\nNn=\iin$ et donc ca s'écrit plus simplement} 

\olivier{Question: Ca peut aussi s'écrire du coup
  autrement. L'écrire. 
}

\olivier{Question: Ca peut aussi s'écrire du coup:
$$\nk^-=\min(\{0 \le \nk \le \nN \mbox{ and } \spsi(\nk \cdot
  \frac{1}{\nN},\comega) \ge -\nepsilon\})$$
Est-ce plus clair ou mieux?
En tout les cas, plus subtile car le $\spsi(\nk \cdot
  \frac{1}{\nN},\comega) \ge - 2\nepsilon$ est en fait en train de parler composante par composante
}


  Function $\ssf$ is continuous, hence uniformly continuous on its
  domain.

Consider $\sx=\inf(\frac{\nk^-}{\nN})$, and $\sm=f(\sx)$. 
Then we claim that $f(\sy) \le \sm$
for all standard $\sy$. Indeed, any $\sy$ contains at least one $\nk' \cdot
  \frac{1}{\nN}$, $0 \le \nk' \le \nN$ infinitely close to it: 
  Consider $\nk'=\lceil \nN
  \cdot \sy \rceil$. 

Hence $\ssf(\sy)$ is $\nepsilon$ close to $\ssf(\nk' \cdot
  \frac{1}{\nN})$. Now, by construction $\ssf(\nk' \cdot
  \frac{1}{\nN}) \le 
\ssf(\nk^- \cdot
  \frac{1}{\nN}) - \nepsilon$, hence $\ssf(\sy) \le
  \sm$. \olivier{Arnaque. Detailler.}

Considering $\sup$, and the $\max(\nS)$
defined symmetrically, provides  a left-computable
zero. 

\end{proof}

\olivier{
section{Admissibility}
\label{def:admiss}

Let us define some specific
well-closed ordinals, the \recursivelyregular ones.

\begin{definition}
  A limit ordinal $\kappa$ is \recursivelyregular iff it closed under
  all $(\infty,\kappa)$-partial recursive functions.
\end{definition}

This is sometimes called admissibility.

The relation of this concept to classical concepts such as the
following concept 
is discussed in \cite{hinman2017recursion}.

\begin{definition}
  A limit ordinal $\alpha$ is admissible if and only if there doesn't
  exist a function $f$ from $\gamma < \alpha$ to $\alpha$ such that:
  \begin{itemize}
  \item $f$ is unbounded (no greatest element in $\alpha$) and
  \item $f$ is $\Sigma_1$-definable in $L_\alpha$.
  \end{itemize}
where $L_\alpha$ is the $\alpha$th level of the G\"odel hierarchy of
constructibles. 
\end{definition}

Relations (equivalence) with other notions of metafinite recursion
theory (e.g. \cite{sacks1990higher}) is also discussed in
\cite{hinman2017recursion}.

Notice that monographies  \cite{hinman2017recursion} and
\cite{sacks1990higher} are both  available online
on Euclid's project. 
}

\section{Discussions and Perspectives}
\label{sec:discuss}

Our presentation of concepts of computable analysis is based on \cheap
analysis. It may be important to discuss how this relates to
other approaches for presenting computable analysis, in particular to
Type-2 analysis.   Type-2 analysis is based on the concept of notation and
representations: A {notation of a
  denumerable set $X$ is a surjective function $\nu$ from a subset of
  $\Sigma^*$ to  $X$, where $\Sigma^*$ is the set of finite words over
  alphabet $\Sigma$.  A representation of a non denumerable set $X$ is
  a surjective function $\delta$ from a subset of $\Sigma^\omega$ to
  $X$ where $\Sigma^\omega$ denotes infinite words over alphabet
  $\Sigma$. Having fixed a representation or a notation for $X$ and for $Y$, a
  function $f$ from $X$ to $Y$ is then considered as computable if it
  has a realizer: given any representation of $x \in X$, the machine
  outputs a representation of $f(x)$. 
  A common representation of $\R$ is Cauchy's representation: a
  real $x \in \R$ is represented by a fastly converging sequence
  $(q_n)_n$ of rationals, that is to say 
  such that $|x-q_n| \le 2^{-n}$. With such a representation
  of $\R$, Type-2 Analysis basically considers machines working over
  sequences of rationals.  

  We want to point out that \cheap analysis brings meaning to
  sequences, as such a sequence of rationals can be read as a \cheap
  rational number.
  However, this analogy is not so direct, as in \cheap analysis,
  sequences are considered as equal if they coincide after some
  \finite rank, contrary to Type-2 analysis where two reals are (considered to
  be) equal  iff they have the same set of representations. Indeed, this does
  not clearly
  imply the existence of a formal simple translation
  from one framework to the other.

  Despite these difficulties, we believe that our framework provides a
  dual view of statements from computable analysis.

We also believe in the pedagogical 
  value of \cheap analysis, in particular when talking about
  computability. More precisely, through this paper, we exposed that several of the concepts
  from Analysis and 
  established properties have very nice presentations in this framework, either
  avoiding quantifier alternations, or relying on simpler to grasp concepts. 

  While some of the presented results are not new, the intended main
  interest of the discussed framework is not in establishing new
  statements but in its elegance. Actually, as most of our
  computability notions are proved to be similar to classical notions,
  if something can be proved using our framework, it can be proved
  using a classical reasoning. This criticism is not a side effect but
  an advantage we take into account to obtain a richer overview of
  some major mathematical concepts. It may also be important to put
  this discussion in the context of the usual criticisms about NSA's
  approach: In particular, NSA transfer property basically implies
  that any result proved in NSA can be proved without NSA, and hence
  this is sometimes used as an argument against this approach: See
  e.g. \cite{wikipediansa} for more deeper discussions and references
  about arguments against NSA's approach. It may also be important to
  do not forget that NSA and \cheap analysis differ, and that \cheap
  analysis approach is of some help to provide constructivity.

  In this paper, we derived our results using $\omega$ as an index
  set. But it could be some other well-ordered and well-closed
  set. This would thus provide some alternative views of statements from
  computability and analysis: In particular, once such a set is fixed,
  previous constructions provide infinitesimal and infinitely large
  elements with respect to all elements in that set. This can hence be
  iterated, using a transfinite induction, to provide richer and
  richer index sets, providing statements in richer and richer
  frameworks. Such an approach yields to contexts such like
  computation models with ordinal times.


As we mentioned before, also notice that links have already been
established between type-2 computability and transfinite computations
(see \cite{GNCompAn} for example) using surreal numbers to extend
$\R$.  More generally, the links between these ideas and already
studied models of computations over the ordinals deserve due
attention.  Models of computation over the ordinals include Sacks'
higher recursion theory \cite{sacks1990higher} or Infinite Time
Register Machines \cite{CarlFKMNW10} or Infinite Time Turing Machines
\cite{infTur}.

\olivier{COMMENTe:

It is now natural to
  wonder how to extend the sequences indexed by natural numbers used
  in \cheap to fit a potential transfinite time simulation. 

  A generalized index for sequences is thus motivated by our
  willingness not only to link computability to cheap non-standard
  analysis (CNSA) as we done above but also to higher-order
  computability, in particular to the generalized computability
  defined along well-closed sets such that the admissible ordinals
  (see Sacks' higher recursion theory for example,
  \cite{sacks1990higher}). It is now generally admitted that
  introducing infinity in computations is a way to better
  characterizing what it means to be computable for different objects,
  by introducing infinity in the underlying set on which computations
  happens (e.g. the reals in analysis, \cite{Wei00}) or on the
  underlying notion of time (e.g. ordinal computations,
  \cite{sacks1990higher,dagBook}). We believe that infinitely large
  numbers introduced in CNSA but extended to indices belonging to some
  well-chosen sets are of some help to establish interesting parallels
  between higher-order recursion models and what is called a kind of
  higher recursion. This leads to further perspectives, for instance a
  new paradigm equivalent, to some extent, to certain existing
  transfinite time computation models, hence introducing a more
  mathematical and thus easier to manipulate model than the current
  ones. Here we allow ourselves some speculation on the possibilities
  offered by the proposed framework to further explore computability
  theory applied to reals and higher-order computability (and their
  relations). 

  We would precisely discuss the case where we consider the index to
  be a \recursivelyregular ordinal $\kappa$ greater than $\N$, an
  ordinal closed under all $(\infty,\kappa)$-partial recursive
  functions. We would hence consider functions from $\kappa$ to
  $\kappa$ instead of functions over the natural numbers. Notice that
  actually, we can even assume that $g(0)=0$, for a successor ordinal
  $\alpha=\beta+1$, we have $g(\beta+1)=g(\alpha)+1$, and for a limit
  ordinal $\alpha$, we have $g(\alpha)>g(\beta)$ for all
  $\beta<\alpha$, to fulfill intuition about infinite largeness. We
  fall in the framework of computability models for such functions.
  Basically, looking at our arguments, we only need a notion of
  computability for \totalrecursive that is stable by composition,
  primitive recursion and minimization, and that include classical
  \totalrecursive functions over $\N$, where such computable functions
  can be enumerated, and with Theorem \ref{patchproperty}. If all of
  this holds, then all previous reasoning are valid for such an
  extension. We could base our work on models that have already been
  proposed, in particular in \cite[Chapter 8]{hinman2017recursion} or
  through Sack's $\alpha$-recursion ( \cite{sacks1990higher}). A first
  guess would be to say that in ordinary recursion theory an object is
  finite iff it is in a one-to-one correspondance with a natural
  number, i.e. with an element of the fundamental domain. Repeating
  \cite{hinman2017recursion}: If this were the only property of finite
  objects considered, given the Axiom of Choice, this would make every
  set \metafinite. It is then fundamental to observe that in addition,
  that in classical recursion theory every finite set of natural
  numbers is in a computable one-to-one correspondance with a natural
  number. In that spirit, an object is termed \metafinite {} iff it is
  in a computable one-to-one relationship with an ordinal
  \cite{hinman2017recursion}.

  Besides, as surreal numbers previously mentioned citing
  \cite{GNCompAn}, another way to consider the question of extensions
  to biggest numbers can be envisaged. To any \cheap number $0<\nx$,
  we can associate $\mu(\nx)=\{ \sy \mbox{ standard } | \sy < \nx\}.$
  As this is a well-ordered set, it has an order type (i.e. an
  ordinal) that we can denote $\nu(\nx)$.  For any \finite $\nx$,
  $\nu(\nx)$ is $\nx$. Previous constructions yield ways to provide
  infinitely large \cheap numbers $\nx$ but from the fact that
  standard numbers live in $\Idx$, for any \cheap number $\nx$, we
  always have $\nu(\nx) \le \kappa$, where $\kappa$ is the order type
  of $\Idx$. In particular, $\nu(\nx) \le \omega$ when $\Idx=\N$.As a
  consequence, considering index sets $\Idx$ with order type $\kappa$
  that are greater than $\omega$ is the only way to get infinitely
  large \cheap numbers with even greater $\nu$ value.  Of course, a
  symmetric discussion is about how small an infinitesimal can be
  (e.g. by reasoning on $\frac{1}{\nx}$).

}

\olivierOK{
V10: Faut de place supprime:
subsection{Effective Preservation Property}

Our previous reasonings actually proves the following general result: 

\begin{theorem}[Effective transfer property]
Let $\nx$ be some \cheap number that satisfies some recursive standard
predicate $\sP(\sn,\sx)$ for all standard $\sn$: 
$$\mbox{
For all standard $\sn$, }
\sP(\sn,\nx) $$

Assume that predicate $\sP$ is monotone:
$\sP(\sn',\sm)$ implies $\sP(\sn,\sm)$ for $\sn' \ge \sn$. 


In the general case, there always exists some \cheap \finite index $\nn$
with $\sP(\comega,\shiftd{\nx}{\nn}).$

When $\nx$ is computable, we have that  $\nn \le \nn_i$ for some $i$,
and there exists some \totalrecursive $\ssf: \N \to \N$ such that $$\mbox{
For all standard $\sn$, }
\sP(\sn, \shiftd{\nx}{\ssf(\sn)}). $$
\end{theorem}

\olivierOK{On est bien d'accord? theorem VERIFIER.}
}

\bibliography{bournez,perso,Sabrina}

\newpage

\begin{acks} 
 This material is based upon work supported by the
  {RACAF Project from Agence National de la
    Recherche}{}. 
Olivier Bournez would like to thank Bruno Salvy for having pointed out the post
\cite{CheapNonStandardAnalysis} when preparing a subject for an exam
  about his logic course.

\end{acks}

\vfill

\olivierOK{Are we ok that all of this is working }
  \sabrinaOK{yes for previous version}


\olivierOK{Important: We said that something is true in the cheap world if it
  holds for any ultrafilter.

Can we say what it would mean for our concepts?

What holds in a particular ultrafilter vs all?
}

\newpage
\appendix

\olivier{Ca suffit comm preuve puor 9.2?}

\end{document}